\documentclass[11pt]{article}

\usepackage{amsmath} 
\usepackage{amsthm} 
\usepackage{amssymb}	
\usepackage{graphicx} 
\usepackage{multicol} 
\usepackage{multirow}
\usepackage{color}
\usepackage[dvips,letterpaper,margin=1in,bottom=1in]{geometry}

\usepackage[utf8]{inputenc}
\usepackage[english]{babel}

\usepackage{diagbox}
\usepackage{mathtools}
\usepackage[export]{adjustbox}

\newtheorem{theorem}{Theorem}[section]

\newtheorem{lemma}[theorem]{Lemma}
\newtheorem{corollary}[theorem]{Corollary}
\newtheorem{proposition}[theorem]{Proposition}
\newtheorem{definition}{Definition}[section]
\newtheorem{example}{Example}
\newtheorem{remark}{Remark}[section]
\newtheorem{problem}{Problem}

\newcommand{\braket}[2]{\left< #1 \vphantom{#2} \middle| #2 \vphantom{#1} \right>} 

\DeclarePairedDelimiter\rbra{\lparen}{\rparen}
\DeclarePairedDelimiter\sbra{\lbrack}{\rbrack}
\DeclarePairedDelimiter\cbra{\{}{\}}
\DeclarePairedDelimiter\abs{\lvert}{\rvert}
\DeclarePairedDelimiter\Abs{\lVert}{\rVert}
\DeclarePairedDelimiter\ceil{\lceil}{\rceil}

\DeclarePairedDelimiter\ket{\lvert}{\rangle}
\DeclarePairedDelimiter\bra{\langle}{\rvert}

\newcommand{\tr} {\operatorname{tr}}
\newcommand{\poly} {\operatorname{poly}}
\newcommand{\polylog} {\operatorname{polylog}}
\newcommand{\rank} {\operatorname{rank}}

\newcommand{\sgn} {\operatorname{sgn}}
\newcommand{\SV} {\text{SV}}

\newcommand{\Real} {\operatorname{Re}}
\newcommand{\Imag} {\operatorname{Im}}

\usepackage{latexsym}
\usepackage{CJK}

\usepackage{enumerate}

\usepackage{algorithm}
\usepackage{algpseudocode}

\usepackage{stmaryrd}
\usepackage{booktabs}

\usepackage{hyperref}
\newcommand{\footremember}[2]{%
    \footnote{#2}
    \newcounter{#1}
    \setcounter{#1}{\value{footnote}}%
}

\usepackage{cite}
\usepackage{bbm}

\usepackage{tablefootnote}
\usepackage{threeparttable}
\usepackage{scrextend}

\usepackage{tikz}
\usetikzlibrary{positioning}

\usepackage{subfig}
\usepackage{quantikz}

\begin{document}

    \title{Fast Quantum Algorithms for Trace Distance Estimation}
        \author{
            Qisheng Wang \footremember{1}{Qisheng Wang is with the Graduate School of Mathematics, Nagoya University, Nagoya, Japan (e-mail: \url{QishengWang1994@gmail.com}).}
            \and Zhicheng Zhang \footremember{2}{Zhicheng Zhang is with the Centre for Quantum Software and Information, University of Technology Sydney, Sydney, Australia (e-mail: \url{iszczhang@gmail.com}).}
        }
        \date{}
        \maketitle

    \begin{abstract}
    In quantum information, trace distance is a basic metric of distinguishability between quantum states.
    However, there is no known efficient approach to estimate the value of trace distance in general. 
    In this paper, we propose efficient quantum algorithms for estimating the trace distance within additive error $\varepsilon$ between mixed quantum states of rank $r$. 
    Specifically, we first provide a quantum algorithm using $r \cdot \widetilde O\rbra{1/\varepsilon^2}$ queries to the quantum circuits that prepare the purifications of quantum states.
    Then, we modify this quantum algorithm to obtain another algorithm using $\widetilde O\rbra{r^2/\varepsilon^5}$ samples of quantum states,
    which can be applied to quantum state certification. 
    These algorithms have query/sample complexities that are independent of the dimension $N$ of quantum states, and their time complexities only incur an extra $O\rbra{\log\rbra{N}}$ factor.
    In addition, we show that the decision version of low-rank trace distance estimation is $\mathsf{BQP}$-complete.
    \end{abstract}

    \textbf{Keywords: 
quantum algorithms, trace distance, singular value decomposition, \\ Hadamard test.}

    \newpage

    \tableofcontents
    \newpage
\section{Introduction}
    
    Distinguishability measures play an important role in quantum computing and quantum information \cite{NC10,Wil13,Wat18}. Trace distance \cite{Hel67,Hel69} and fidelity \cite{Uhl76,Joz94} are two of the most commonly employed distinguishability measures between quantum states, which also have generalizations to quantum channels \cite{Kit97,Rag01,GLN05,PMM07} and quantum strategies \cite{CDP08,CDP09,Gut12,GRS18}.
    
    The trace distance 
    between two mixed quantum states $\rho$ and $\sigma$ is a metric, defined by
    \begin{equation}
        T \rbra*{\rho, \sigma} = \frac 1 2 \tr\rbra*{\abs*{\rho - \sigma}}.
    \end{equation}
    Compared to fidelity, trace distance has an operational interpretation for the maximum success probability in distinguishing quantum states in a quantum hypothesis testing experiment \cite{Hel69}.
    Estimating the value of trace distance is a basic problem both in practice and in theory. 
    
    A large amount of efforts (cf.\ \cite{RASW23}) have been made to estimate trace distance and fidelity.
    Classically, 
    they can be computed through semidefinite programming \cite{Wat09,Wat13} with time complexity polynomial in the dimension of the quantum states, which however grows exponentially as the number of qubits increases.
    By contrast, the fidelity between pure quantum states can be efficiently estimated by the SWAP test \cite{BCWdW01}.
    Generally, trace distance and fidelity estimation is even hard on quantum computers, as is shown in \cite{Wat02,Wat09b} to be $\mathsf{QSZK}$-hard.
    Nevertheless, several approaches were proposed for estimating the trace distance \cite{ZRC19,ZR22} and fidelity \cite{TYKI06,GLGP07,GT09,FL11,dSLCP11} in some practical scenarios, including those using variational quantum algorithms \cite{CPCC20,TV21,LLSL21,CSZW22,RASW23}.
    
    Since it was found that low-rank quantum states can be reconstructed with significantly fewer samples and measurements \cite{GLF+10,FGLE12,OW16,HHJ+17,vACGN22} than by general quantum state tomography \cite{DM97,DMP03}, the closeness between low-rank quantum states has attracted extensive attention. 
    For example, quantum state certification with respect to trace distance and fidelity was investigated in \cite{BOW19}, where the low-rank case was considered. 
    Recently, a polynomial-time quantum algorithm for estimating the fidelity of low-rank quantum states was developed in \cite{WZC+22}, which was later improved by \cite{WGL+22,GP22}.
    Inspired by them, a quantum algorithm for estimating the trace distance of low-rank quantum states was then developed in \cite{WGL+22},
    which proves the conjecture proposed in \cite{CCC19} that low-rank trace distance estimation is in $\mathsf{BQP}$.
    However, these known quantum algorithms for estimating the trace distance and fidelity of low-rank quantum states mentioned above have large exponents of rank and precision in their time complexities (see Table \ref{tab:complexity} for comparison). 

    In this paper, we consider the low-rank trace distance estimation problem, stated as follows.

    \begin{problem} [Low-rank trace distance estimation]
    \label{prb:main}
    Given two $N$-dimensional mixed quantum states $\rho$ and $\sigma$ of rank $r$, 
    the task is to estimate $T\rbra{\rho, \sigma}$ within additive error $\varepsilon$.
    \end{problem}
    
    \begin{table*}[!htp]
\centering
\caption{Complexity of trace distance estimation and fidelity estimation.}
\normalsize
\label{tab:complexity}
\begin{threeparttable}
\begin{tabular}{cccc}
\toprule
Task            & Resources        & Query/Sample Complexity        & Approach                                     \\ \midrule
Tomography             & Purified Access & $\widetilde O\rbra{Nr/\varepsilon}$\tnote{*} & \cite{vACGN22}     \\ \cmidrule{2-4}
                    & Identical Copies & $\widetilde \Theta\rbra{Nr/\varepsilon^2}$ & \cite{OW16,HHJ+17}     \\ \midrule
Trace Distance              & Purified Access  & $\widetilde O\rbra{r^{5}/\varepsilon^{6}}$                              & \cite{WGL+22}      \\ \cmidrule{3-4}
            &                  & $r \cdot \widetilde O \rbra{1/\varepsilon^{2}}$                                  & {Algorithm \ref{algo:purified}}        \\  \cmidrule{2-4}
                    & Identical Copies & $\widetilde O\rbra{r^2/\varepsilon^{5}}$                                & {Algorithm \ref{algo:sample}}         \\ \midrule
Fidelity            & Purified Access  & $\widetilde O\rbra{r^{12.5}/\varepsilon^{13.5}}$                        & \cite{WZC+22}      \\ \cmidrule{3-4}
                    &                  & $\widetilde O\rbra{r^{6.5}/\varepsilon^{7.5}}$                          & \cite{WGL+22}      \\ \cmidrule{3-4}
                    &                  & $\widetilde O\rbra{r^{2.5}/\varepsilon^{5}}$                            & \cite{GP22}        \\ \cmidrule{2-4}
                    & Identical Copies & $\widetilde O\rbra{r^{5.5}/\varepsilon^{12}}$                           & \cite{GP22}        \\
                    \bottomrule
\end{tabular}
\begin{tablenotes}
\item[*] Here, $N$ is the dimension of quantum states, $r$ is the rank of quantum states, and $\varepsilon$ is the required additive precision. 
\end{tablenotes}
\end{threeparttable}
\end{table*}
    
    \subsection{Main Results}
    
    We propose two efficient quantum algorithms for low-rank trace distance estimation: 
    \begin{itemize}
        \item {Algorithm \ref{algo:purified}} (purified access) with query complexity $r \cdot \widetilde O\rbra{1/\varepsilon^2}$ (see Corollary \ref{corollary:purified}); and
        \item {Algorithm \ref{algo:sample}} (sample access) with sample complexity $\widetilde O\rbra{r^2/\varepsilon^5}$ (see Corollary \ref{corollary:sample}). 
    \end{itemize}
    Here, $\widetilde O\rbra{f\rbra{a, b}} = O\rbra{f\rbra{a, b} \polylog\rbra{a, b}}$ suppresses polylogarithmic factors of parameters that appear in $\widetilde O\rbra{\cdot}$. 
    Both algorithms have small exponents of rank and precision in their complexities, thus are more suitable to be implemented in practice.
    They are also time-efficient in the sense that they have the same quantum time complexities as their query/sample complexities up to a logarithmic factor of $N$. 
    We compare them with known approaches in Table \ref{tab:complexity}, and discuss their implications in the following. 
    
    \subsubsection{Purified Access}
    Suppose that we are given quantum circuits $O_\rho$ and $O_\sigma$ preparing the purifications of $N$-dimensional mixed quantum states $\rho$ and $\sigma$. Specifically, 
    \begin{align}
        \ket{\rho}_{n+n_\rho} & = O_\rho \ket{0}_n \ket{0}_{n_\rho}, \\
        \ket{\sigma}_{n+n_\sigma} & = O_\sigma \ket{0}_n \ket{0}_{n_\sigma},
    \end{align}
    where $N = 2^n$, and the subscripts $n$, $n_\rho$ and $n_\sigma$ indicate not only the subspace but also the number of qubits involved. 
    We assume that $n_\rho, n_\sigma \leq n$ for simplicity. 
    Then, $\rho$ and $\sigma$ are obtained by tracing out the ancilla qubits:
    \begin{align}
        \rho & = \tr_{n_\rho}\rbra*{\ket{\rho}_{n+n_\rho}\bra{\rho}}, \\
        \sigma & = \tr_{n_\sigma}\rbra*{\ket{\sigma}_{n+n_\sigma}\bra{\sigma}}.
    \end{align}
    This input model, known as the quantum purified access model, is commonly used in quantum computational complexity and quantum algorithms \cite{Wat02,BKL+19,vAG19,GL20,GLM+22,RASW23,GHS21,SH21}. 
    
    Our first result is {Algorithm \ref{algo:purified}}, given purified access to the input quantum states (i.e., quantum circuits that prepare their purifications), with query complexity $r \cdot \widetilde O\rbra{1/\varepsilon^2}$. This achieves a linear dependence on the rank $r$ in the time complexity, compared to the prior best $\widetilde O\rbra{r^5/\varepsilon^6}$ by \cite{WGL+22}.

    Note that for pure quantum states, i.e., $r = 1$, trace distance can also be computed by the identity (cf. \cite[Equation (9.173)]{Wil13})
    \begin{equation} \label{eq:trace-distance-by-fidelity-pure}
        T\rbra*{\ket{\psi}, \ket{\phi}} = \sqrt{1 - \rbra*{F\rbra*{\ket{\psi}, \ket{\sigma}}}^2},
    \end{equation}
    where $F\rbra{\ket{\psi}, \ket{\sigma}} = \abs{\braket{\psi}{\phi}}$ is the fidelity between $\ket{\psi}$ and $\ket{\phi}$. 
    Suppose that $U_\psi$ and $U_\phi$ are quantum circuits that prepare $\ket{\psi}$ and $\ket{\phi}$, respectively;
    then we can estimate $T\rbra{\ket{\psi}, \ket{\phi}}$ within additive error $\varepsilon$ using $O\rbra{1/\varepsilon^2}$ queries to $U_\psi$ and $U_\phi$ by the SWAP test \cite{BCWdW01} (or {\cite[Algorithm 1]{RASW23}}) equipped with quantum amplitude estimation \cite{BHMT02} (see Appendix \ref{app:estimation-via-swap} for details).
    By comparison, {Algorithm \ref{algo:purified}} has the same complexity (only up to a logarithmic factor), and retains the $\varepsilon$-dependence even when quantum states are not pure.
    
    \subsubsection{Sample Access} 
    Suppose that identical copies of $\rho$ and $\sigma$ are directly given. 
    Our second result is {Algorithm \ref{algo:sample}}, given identical copies, with sample complexity $\widetilde O\rbra{r^2/\varepsilon^5}$, while no prior explicit
    sample complexity is known for this task. This is done by modifying {Algorithm \ref{algo:purified}} via the technique of density matrix exponentiation \cite{LMR14,KLL+17}, inspired by \cite{GP22}.
    
    A related problem --- quantum state certification with respect to trace distance given identical copies was studied in \cite{BOW19} (see also  \cite{GL20} for the case of purified access), which is to distinguish between the cases $T\rbra{\rho, \sigma} = 0$ or $T\rbra{\rho, \sigma} > \varepsilon$ with a promise that it is in either case.
    For low-rank quantum states, the sample complexity of state certification was shown in 
    \cite{BOW19} to be $\Theta\rbra{r/\varepsilon^2}$.
    Note that low-rank state certification can be
    solved by low-rank trace distance estimation; 
    however, it is not known whether the  
    converse is possible. 
    {Algorithm \ref{algo:sample}} implies a quantum algorithm, given identical copies, for low-rank state certification with time complexity $\widetilde O\rbra{r^2/\varepsilon^5 \cdot \log\rbra{N}}$, compared to the approach by \cite{BOW19} with time complexity $\widetilde O\rbra{ r^3/\varepsilon^6 + r/\varepsilon^2 \cdot \log\rbra{N} }$ (as noted in \cite{Wri22}, this is obtained by weak Schur sampling, cf.\ \cite{MdW16}, with the best known quantum Fourier transform over symmetric groups \cite{KS16}), though with a slightly higher sample complexity than \cite{BOW19}. 
    We compare them in Table \ref{tab:cmp-certification}. 
    
    \begin{table*}[!htp]
    \centering
    \caption{Sample/time tradeoff for quantum state certification with respect to trace distance.}
    \small
    \label{tab:cmp-certification}
    \begin{tabular}{cccc}
    \toprule
    Task          & Sample Complexity                & Time Complexity      & Approach        \\ \midrule
    Estimation    & $O\rbra*{r^2/\varepsilon^5 \cdot \log^2\rbra{r/\varepsilon} \log^2\rbra{1/\varepsilon} }$     & $O\rbra*{r^2/\varepsilon^5 \cdot \log^2\rbra{r/\varepsilon} \log^2\rbra{1/\varepsilon} \log\rbra{N} }$ & {Algorithm \ref{algo:sample}} \\ \midrule
    Certification & $\Theta \rbra*{r/\varepsilon^2}$ & $O\rbra*{r^3/\varepsilon^6 \cdot \log\rbra{r/\varepsilon} + r/\varepsilon^2 \cdot \log\rbra{N} }$ & \cite{BOW19} \\ \bottomrule
    \end{tabular}
    \end{table*}
    
    \subsection{Technical Overview}
    
    We give high-level overview of our quantum algorithms with both purified access and sample access. 
    
    \subsubsection{Purified Access}
    
    The prior best quantum algorithm for low-rank trace distance estimation is by \cite{WGL+22}, with query complexity $\widetilde O\rbra*{r^5/\varepsilon^6}$. In their approach, the key observation is the identity
    \begin{equation} \label{eq:trace-distance-WGL+22}
        T\rbra*{\rho, \sigma} = \tr\rbra*{ \sqrt{\abs*{\nu_{-}}} \Pi_{\nu_{+}} \sqrt{\abs*{\nu_{-}}} },
    \end{equation}
    where $\nu_{\pm} = \rbra*{\rho \pm \sigma} / 2$ and $\Pi_{\varrho}$ denotes the projector onto the support subspace of $\varrho$. Their idea, roughly speaking, is to prepare a quantum state block-encoding of $\sqrt{\abs*{\nu_{-}}} \Pi_{\nu_{+}} \sqrt{\abs*{\nu_{-}}}$ by performing a unitary block-encoding of $\sqrt{\abs*{\nu_{-}}}$ on a quantum state block-encoding of $\Pi_{\nu_{+}}$; then estimate the trace of the resulting quantum state following Eq. \eqref{eq:trace-distance-WGL+22} by quantum amplitude estimation \cite{BHMT02}. 
    This algorithm is inefficient mainly because it employs square roots of semidefinite operators and a heavily nested structure, which take considerable computational costs.
    
    To overcome these issues, we provide an efficient quantum algorithm for low-rank trace distance estimation, which is, technically, very different from the one given in \cite{WGL+22} just mentioned. We still use the notations above, and consider the singular value decomposition $\nu_- = W \Sigma V^\dag$. 
    Then, the trace distance can be expressed by the following identity:
    \begin{equation} \label{eq:trace-distance-sgn}
        T\rbra*{\rho, \sigma} = \frac 1 2 \rbra[\Big]{ \tr\rbra*{\rho \sgn^{\SV}\rbra*{\nu_-}} - \tr\rbra*{\sigma \sgn^{\SV}\rbra*{\nu_-}} },
    \end{equation}
    where $\sgn^{\SV}\rbra{\nu_-} = W \sgn\rbra{\Sigma} V^\dag$ is the singular value transformation of $\nu_-$ by the sign function
    \begin{equation}
        \sgn\rbra*{x} = \begin{cases}
            1, & x > 0, \\
            0, & x = 0, \\
            -1, & x < 0.
        \end{cases}
    \end{equation}
    To see this, we note that because $\nu_-$ is Hermitian, the spectral decomposition $\nu_- = U \Lambda U^\dag = U \abs{\Lambda} \sgn\rbra{\Lambda} U^\dag$ is also a singular value decomposition with $W = U$, $\Sigma = \abs{\Lambda}$, and $V = U \sgn\rbra{\Lambda}$ (cf. \cite[Theorem 5.5]{LB97}). Then, $\sgn^{\SV}\rbra{\nu_-} = U \Pi_{\nu_-} \sgn\rbra{\Lambda} U^\dag$, where $\Pi_{\nu_-}$ is the projector onto the support of $\nu_-$. By a simple calculation, it can be verified that $\nu_- \sgn^{\SV}\rbra{\nu_-} = U \abs{\Lambda} U^\dag$, which gives $\tr\rbra{\nu_- \sgn^{\SV}\rbra{\nu_-}} = \tr\rbra{\abs{\Lambda}} = \tr\rbra{\abs{\nu_-}} = T\rbra{\rho, \sigma}$.
    This allows us to estimate the values of $\tr\rbra{\rho \sgn^{\SV}\rbra{\nu_-}}$ and $\tr\rbra{\sigma \sgn^{\SV}\rbra{\nu_-}}$ separately, which can be done by combining the QSVT (quantum singular value transformation) technique \cite{GSLW19} with the Hadamard test \cite{AJL09}, inspired by \cite{GP22}. 
    
    To give an intuitive overview of our algorithm, the main idea is that $\tr\rbra{\rho \sgn^{\SV}\rbra{\nu_-}}$ can be estimated by the Hadamard test with an (approximate) unitary block-encoding of $\sgn^{\SV}\rbra{\nu_-}$ and the quantum state $\rho$, as shown in Figure \ref{fig:explanation}.
    
    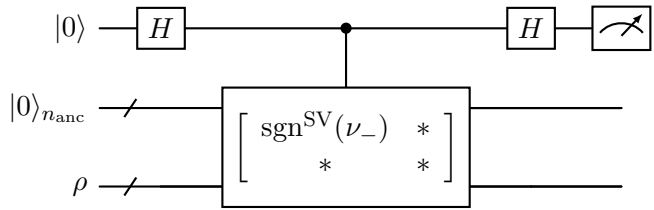
\begin{figure}[!htp]
    \centering
    \begin{quantikz}[row sep=0.5cm, ampersand replacement=\&]
        \lstick{$\ket{0}$}                    \& \gate{H} \& \ctrl{1}                \& \gate{H} \& \meter{} \\
        \lstick{$\ket{0}_{n_{\textup{anc}}}$} \& \qwbundle{}     \& \gate[2, disable auto height]{\left[ 
        \begin{array}{cc} \sgn^{\SV}\rbra*{\nu_-} & * \\ * & * \end{array} 
        \right]} \& \qw \&\qw \\
        \lstick{$\rho$}                       \& \qwbundle{}      \& \qw                        \& \qw \& \qw\\
    \end{quantikz}
    \caption{Hadamard test for estimating $\tr\rbra{\rho \sgn^{\SV}\rbra{\nu_-}}$, 
    which gets measurement outcome $0$ with probability $\rbra{1 + \tr\rbra{\rho \sgn^{\SV}\rbra{\nu_-}}}/2$, where $n_{\textup{anc}}$ 
    is the number of ancilla qubits.}
    \label{fig:explanation}
    \end{figure}
    
    \subsubsection{Sample Access}
    Our quantum algorithm with purified access is specifically designed so that it can be modified at only a little cost to obtain another algorithm that only uses identical copies. 
    We note that in {Algorithm \ref{algo:purified}}, purified access is only used for: 
    \begin{enumerate}
        \item Constructing unitary block-encodings $U_\rho$ and $U_\sigma$ of $\rho$ and $\sigma$, respectively; and
        \item Preparing identical copies of $\rho$ and $\sigma$ for the Hadamard test. 
    \end{enumerate}
    Actually, the two types of demands are also achievable with only identical copies. 
    The first demand can be achieved by density matrix exponentiation \cite{LMR14,KLL+17}, which was recently employed in \cite{GP22} to develop quantum algorithms for fidelity estimation;
    and the second demand is without doubts because identical copies are directly given. 

    As will be shown in {Algorithm \ref{algo:sample}}, density matrix exponentiation \cite{LMR14,KLL+17} is only used to produce quantum channels that approximately implement the unitary block-encodings $U_\rho$ and $U_\sigma$ constructed in {Algorithm \ref{algo:purified}}.
    Technically, we still need to (approximately) implement their inverses $U_\rho^\dag$ and $U_\sigma^\dag$. 
    To resolve this issue, suppose that $U_\rho$ is approximately implemented by the quantum channel $\mathcal{E}$ such that $\Abs{\mathcal{E} - U_\rho}_\diamond \leq \delta$, and $\mathcal{E}$ is given by a quantum circuit $W$ with $k$ samples of $\rho$ such that
    \begin{equation} \label{eq:approx-implement}
        \mathcal{E}\rbra*{\varrho} = \tr_{\text{env}}\rbra*{ W \rbra*{ \underbrace{\rho^{\otimes k} \otimes \ket{0}_\ell\bra{0}}_{\text{env}} \otimes \varrho} W^\dag }
    \end{equation}
    for every quantum state $\varrho$. 
    Then, consider the quantum channel $\mathcal{E}^{\textup{inv}}$ that is obtained by using $W^\dag$ in place of $W$ in Eq. (\ref{eq:approx-implement}).
    It can be shown that $\mathcal{E}^{\textup{inv}}$ approximately implements $U_\rho^\dag$, i.e., $\Abs{\mathcal{E}^{\textup{inv}} - U_\rho^\dag}_\diamond \leq \delta$ 
    (see Lemma \ref{lemma:block-encoding-invertible}). 
    This is visualized in Figure~\ref{fig:U-to-Udag}. 
    
    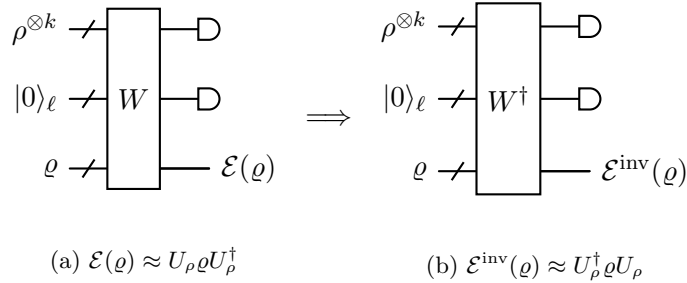
\begin{figure}[!htp]
        \centering
        \subfloat[$\mathcal{E}\rbra*{\varrho} \approx U_\rho \varrho U_\rho^\dag$]{
        \begin{quantikz}
            \lstick{$\rho^{\otimes k}$}            & \gate[3]{W}\qwbundle{}   & \meterD{} \\
            \lstick{$\ket{0}_{\ell}$}      &  \qwbundle{}               & \meterD{} \\
            \lstick{$\varrho$}                    &    \qwbundle{}            &\qw\rstick{$\mathcal{E}(\varrho)$} \\
        \end{quantikz}
        }
        \subfloat{
        $\begin{aligned}\Longrightarrow\end{aligned}$
        }
        \renewcommand\thesubfigure{b}\subfloat[$\mathcal{E}^{\textup{inv}}\rbra*{\varrho} \approx U_\rho^\dag \varrho U_\rho$]{
        \begin{quantikz}
            \lstick{$\rho^{\otimes k}$}            & \gate[3]{W^\dag}\qwbundle{}   & \meterD{} \\
            \lstick{$\ket{0}_{\ell}$}      &  \qwbundle{}               & \meterD{} \\
            \lstick{$\varrho$}                    &    \qwbundle{}            &\qw\rstick{$\mathcal{E}^{\textup{inv}}(\varrho)$} \\
        \end{quantikz}
        }
        \caption{Quantum circuit for approximately implementing the inverse of unitary operators.}
        \label{fig:U-to-Udag}
    \end{figure}

    We can use the above method to approximate $U_\rho$, $U_\sigma$ and their inverses and controlled versions by quantum channels that use only samples of $\rho$ and $\sigma$. 
    This allows us to extend our algorithm with purified access to the case of sample access.
    The main step, shown in Figure \ref{fig:explanation}, is to estimate $\tr\rbra{\rho \sgn^{\SV}\rbra{\nu_-}}$ by the Hadamard test with a unitary block-encoding $U_{\sgn^{\SV}\rbra{\nu_-}}$ of $\sgn^{\SV}\rbra{\nu_-}$.
    We use the QSVT technique to construct $U_{\sgn^{\SV}\rbra{\nu_-}}$ using queries to $U_\rho$ and $U_\sigma$ (and their inverses and controlled versions), which can be approximately implemented by the quantum channels constructed above. 
    This gives us a quantum algorithm for trace distance estimation with sample access. 

    \subsection{Lower Bounds and Hardness}
    
    As our algorithm with identical copies for trace distance estimation also applies to quantum state certification with respect to trace distance, the lower bound for the sample complexity of trace distance estimation follows from that of state certification, which is known to be $\Omega \rbra{r/\varepsilon^2}$ by \cite{BOW19}. 
    The best known lower bound for the time complexity of low-rank trace distance estimation is $\omega\rbra*{\poly\rbra*{\log\rbra*{r}, 1/\varepsilon}}$ unless $\mathsf{BQP} = \mathsf{QSZK}$ by \cite{WGL+22}. 
    However, there is no known lower bound for the query complexity of trace distance estimation. 

    Low-rank fidelity estimation is known to be $\mathsf{BQP}$-complete by combining the $\mathsf{BQP}$-hardness for pure-state fidelity estimation in \cite{RASW23} and the polynomial-time quantum algorithm for low-rank fidelity estimation in \cite{WZC+22}.
    In this paper, we show that low-rank trace distance estimation is also $\mathsf{BQP}$-complete by reducing it from the pure-state fidelity estimation \cite{RASW23} (see Theorem \ref{thm:td-est-bqp}).
    Therefore, there is probably no efficient classical algorithm for low-rank trace distance estimation unless $\mathsf{BQP} = \mathsf{BPP}$.
    Nevertheless, this does not rule out the possibility of dequantized algorithms for low-rank trace distance estimation, if ``sampling and query access'' \cite{Tan19,GLT18,CGL+20,Tan21} to the matrix representations of quantum states is given.
    
    \subsection{Discussion and Extensions}
    
    In this paper, we propose efficient quantum algorithms for low-rank trace distance estimation. 
    This is done by using the formula Eq. \eqref{eq:trace-distance-sgn}, different from prior approaches, that expresses trace distance in two terms and enables us to compute each term separately by combining QSVT \cite{GSLW19} with the Hadamard test \cite{AJL09}. 
    Unlike prior quantum algorithms that take advantage of the low-rank condition \cite{WZC+22,WGL+22,GP22}, we avoid techniques with heavy computational costs such as positive powers of quantum operators. 
    This is the main reason why we are able to achieve a linear dependence on the rank $r$ in the time complexity, thereby yielding a quantum algorithm with sample (and also time) complexity $\widetilde O\rbra{r^2/\varepsilon^5}$ with small exponents of $r$ and $\varepsilon$. 
    
    In real experiments, especially in the NISQ (noisy intermediate-scale quantum) era \cite{Pre18}, the true quantum states are only approximately low-rank. 
    It can be shown that our quantum algorithms apply to not only strictly but also approximately low-rank quantum states, in the sense that the sum of the largest eigenvalues is close to $1$ (see Section \ref{sec:approx-low-rank} for details). 
    By contrast, the quantum algorithm for trace distance estimation in \cite{WGL+22} does not consider this case. The quantum state certification with respect to trace distance studied in \cite{BOW19} considers the approximately low-rank case but does not apply to our estimation task.

    The depth complexity is also an important consideration when designing quantum algorithms, especially in the near-future \cite{BGK18}. 
    Some tasks are known to have low-depth quantum algorithms, e.g., quantum Fourier transform \cite{CW00}, hidden linear function problem \cite{BGK18}, Hamiltonian simulation \cite{ZWY21}, quantum state preparation \cite{STY+21,Ros21,ZLY22}, and multivariate trace estimation \cite{QKW22}.
    Our quantum algorithm given identical copies with sample complexity $\widetilde O\rbra{r^2/\varepsilon^5}$ can be partially parallelized to achieve a depth complexity of $\widetilde O\rbra{r^2/\varepsilon^3}$.
    In comparison, the algorithm for quantum state certification with respect to trace distance in \cite{BOW19} has time complexity $\widetilde O\rbra{r^3/\varepsilon^6}$, while its depth remains the same as its time complexity due to the use of the quantum Fourier transform over symmetric groups in \cite{KS16} (with depth complexity discussed in \cite[Section 5]{KS16}). 
    
    As discussed above, our quantum algorithms are not only efficient in the sense of query/sample and time complexity but also robust to small errors in the input quantum states. 
    For this reason, we believe our algorithms could have potential applications in practice. 
    We hope our techniques in this paper can bring new ideas to other quantum algorithms. 

    \subsection{Recent Developments} 
    After the work described in this paper, several quantum algorithms for quantum state testing regarding trace norm/distance were developed
    based on the framework for trace distance estimation proposed in this paper. 
    \begin{itemize}
        \item Nuradha, Goldfeld, and Wilde \cite{NGW23} proposed a hypothesis testing based auditing pipeline for quantum differential privacy with domain knowledge, and related its type-I error to the number of samples.
        \item 
        Le Gall, Liu, and Wang \cite{LGLW23} showed that the space-bounded decision version of trace distance estimation is $\mathsf{BQL}$-complete (while its time-bounded counterpart, \textit{quantum state distinguishability}, is $\mathsf{QSZK}$-complete \cite{Wat02,Wat09b}), 
        and the space-bounded version of \textit{quantum state certification} \cite{BOW19} is $\mathsf{coRQ_{\text{U}}L}$-complete. 
        Here, the two complexity classes $\mathsf{BQL}$ \cite{Wat03,FR21} and $\mathsf{coRQ_{\text{U}}L}$ \cite{Wat01} stand for Bounded-error Quantum Logarithmic-space and the complement class of $\mathsf{RQ_{\text{U}}L}$ (Randomized Quantum Unitary Logarithmic-space), respectively. 
    \end{itemize}

    \subsection{Open Problems}

    To conclude this section, we mention several open problems that are related to our work:
    \begin{itemize}
        \item Certain tasks for quantum state learning can be done in sublinear time, e.g., von Neumann entropy estimation with multiplicative error \cite{GHS21}.
        It would be interesting to study whether trace distance and fidelity can be estimated (to an additive or multiplicative error) with a \textit{sublinear} dependence on the rank $r$. 
        \item Can we estimate trace distance with shallower quantum circuits?
        \item To the best of our knowledge, the quantum lower and upper bounds for estimating both trace distance and fidelity are far from being tight.
        For example, for sample access, the sample upper bounds for trace distance and fidelity are $\widetilde O\rbra{r^2/\varepsilon^5}$ (this work) and $\widetilde O\rbra{r^{5.5}/\varepsilon^{12}}$ \cite{GP22}, while their corresponding lower bounds are $\Omega\rbra{r/\varepsilon^2}$ and $\Omega\rbra{r/\varepsilon}$ according to \cite{BOW19}.
        A direct question is whether we can tighten the gap between their lower and upper bounds. 
        \item Recently, our quantum algorithms for trace distance estimation were adapted for auditing quantum differential privacy in \cite{NGW23}. 
        Can we find other applications?
    \end{itemize}
    
    \subsection{Organization of This Paper}
    
    In the rest of this paper, we first include necessary preliminaries in Section \ref{sec:preliminary}.
    Then, we will provide quantum algorithms with purified access and identical copies with their analysis, respectively, in Section \ref{sec:algo}. 
    In Section \ref{sec:bqp-completeness}, we show the $\mathsf{BQP}$-completeness of low-rank trace distance estimation. 
    In Section \ref{sec:approx-low-rank}, we consider how our algorithms can be applied to approximately low-rank quantum states.
    
    \section{Preliminaries} \label{sec:preliminary}
    
    In this section, we will introduce quantum query complexity, approximate rank, block-encoding, quantum singular value transformation, and the technique for sampling to block-encoding that will be used in our algorithms. 
    Throughout this paper, we denote $\sbra{N} = \cbra{1, 2, \dots, N}$.

    \subsection{Quantum Query Complexity}
    
    Suppose $\mathcal{O}$ is a quantum unitary oracle (which can be understood as a given quantum circuit). 
    A quantum query algorithm $\mathcal{A}$ can be described by a quantum circuit that consists of \mbox{(controlled-)$\mathcal{O}$} and (controlled-)$\mathcal{O}^\dag$ and elementary quantum gates.
    Throughout this paper, one query to $\mathcal{O}$ means one query to (controlled-)$\mathcal{O}$ or (controlled-)$\mathcal{O}^\dag$ if not specified. 
    The query complexity of $\mathcal{A}$ is the number of queries to $\mathcal{O}$ in $\mathcal{A}$. 
    The time complexity of $\mathcal{A}$ is the number of queries to $\mathcal{O}$ and elementary quantum gates in $\mathcal{A}$. 
    The depth complexity of $\mathcal{A}$ is the maximal length of a (directed) path from an input qubit to an output qubit, where each elementary quantum gate or query to $\mathcal{O}$ costs $1$ unit of length. 
    
    \subsection{Approximate Rank} \label{sec:approx-rank}
    
    Suppose $A = \sum_j \lambda_j \ket{\psi_j} \bra{\psi_j}$ is an Hermitian operator. Let $\rank_\delta\rbra*{A}$ be the approximate rank of $A$ with respect to $\delta$ defined by
    \begin{equation}
        \rank_\delta\rbra*{A} = \sum_{j \colon \abs*{\lambda_j} > \delta} 1.
    \end{equation}
    Especially, the rank of $A$ is $\rank\rbra*{A} = \rank_0\rbra*{A}$. 
    We will discuss how our quantum algorithms can be applied to approximately low-rank quantum states in Section \ref{sec:approx-low-rank}.
    Let $w\rbra*{A, \delta}$ be the sum of absolute eigenvalues of $A$ not greater than $\delta$, defined by
    \begin{equation} \label{eq:w}
        w\rbra*{A, \delta} = \sum_{j \colon \abs*{\lambda_j} \leq \delta} \abs*{\lambda_j},
    \end{equation}
    which will be used in the conditions of our quantum algorithms (see Theorem \ref{thm:purified} and Theorem \ref{thm:sample}). 
    In the following, we give an upper bound for $w\rbra*{A, \delta}$ by $\rank\rbra*{A}$.
    
    \begin{proposition}
        For every $\delta \geq 0$, we have $w\rbra*{A, \delta} \leq \delta \cdot \rank\rbra*{A}$ for every Hermitian operator $A$.
    \end{proposition}
    \begin{proof}
        Suppose that $A$ is $N$-dimensional and $r = \rank\rbra{A}$. 
        Let $\lambda_j$ for $1 \leq j \leq N$ sorted by their absolute values, i.e., $\abs{\lambda_1} \geq \abs{\lambda_2} \geq \dots \geq \abs{\lambda_N}$. 
        Since there are at most $r$ non-zero eigenvalues $\lambda_j$ of $A$, we have $\lambda_j = 0$ for every $r < j \leq N$. 
        Therefore, we conclude that
        \begin{equation}
            w\rbra{A, \delta} = \sum_{j \in \sbra{r} \colon \abs{\lambda_j} \leq \delta} \abs{\lambda_j} \leq \sum_{j \in \sbra{r}} \delta = r \delta = \delta \cdot \rank\rbra{A}.
        \end{equation}
    \end{proof}
    
    \subsection{Quantum Amplitude Estimation}
    
    Estimating the amplitude of a pure quantum state is a basic subroutine that is commonly used in quantum algorithms. 
    
    \begin{theorem} [Quantum amplitude estimation {\cite[Theorem 12]{BHMT02}}] \label{thm:amp-estimation}
        Suppose $U$ is a unitary operator such that
        \begin{equation}
            U \ket{0}\ket{0} = \sqrt{p} \ket{0} \ket{\phi_0} + \sqrt{1 - p} \ket{1} \ket{\phi_1},
        \end{equation}
        where $\ket{\phi_0}$ and $\ket{\phi_1}$ are normalized pure quantum states, and $p \in \sbra*{0, 1}$. 
        There is a quantum algorithm that outputs $\widetilde p$ such that 
        \begin{equation}
            \abs*{\widetilde p - p} \leq \frac{2\pi\sqrt{p(1-p)}}{M} + \frac{\pi^2}{M^2}
        \end{equation}
        with probability $\geq 8/\pi^2$ using $O\rbra*{M}$ queries to $U$. 
    \end{theorem}
    Especially, if no prior knowledge is known for $p$, we can estimate $p$ within additive error $\varepsilon$ using $O\rbra*{1/\varepsilon}$ queries to $U$. 
    \subsection{Block-Encodings}

    Block-encoding is a conventional description of quantum operators (cf. \cite{GSLW19}) when we focus on a certain part (e.g., upper-left corner) of the operators. 
    In this paper, we write $\ket{0}_a$ to denote $\ket{0}^{\otimes a}$, where the subscript $a$ indicates the number of qubits.
    
    \begin{definition} [Block-encoding] \label{def:block-encoding}
    Suppose $A$ is an $n$-qubit operator, $\alpha, \varepsilon \geq 0$ and $a \in \mathbb{N}$. An $(n+a)$-qubit unitary operator $B$ is said to be an $(\alpha, a, \varepsilon)$-block-encoding of $A$, if
    \begin{equation}
        \Abs{ \alpha \prescript{}{a}{\bra 0} B \ket 0_a - A } \leq \varepsilon,
    \end{equation}
    where $\Abs{\cdot}$ denotes the operator norm, defined by
    \begin{equation}
        \Abs*{A} = \sup\limits_{\Abs*{\ket{\psi}} = 1} \Abs*{A \ket{\psi}}.
    \end{equation}
    \end{definition}

    Intuitively, $A$ is represented by the matrix in the upper left corner of $B$, i.e.
    \begin{equation}
        B \approx \begin{bmatrix}
            A/\alpha & * \\
            * & *
        \end{bmatrix}.
    \end{equation}
    
    \subsubsection{Linear Combination of Block-Encoded Operators}

    We will introduce the LCU (Linear-Combination-of-Unitaries) technique \cite{CW12,BCC+15}. 
    LCU is usually used to implement a block-encoding of a linear combination of several block-encoded matrices. 
    Roughly speaking, suppose we are given several unitary oprators $U_k$, where each $U_k$ is a block-encoding of a matrix $A_k$.
    Then, the LCU technique allows us to implement a block-encoding of $\sum_{k} y_k A_k$, which is a linear combination of $A_k$.
    
    The following version of LCU is taken from \cite{GSLW19}.
    
    \begin{definition} [State preparation pair] \label{def:state-preparation-pair}
        Let $y \in \mathbb{C}^m$ with $\Abs{y}_1 \leq \beta$, and $\varepsilon \geq 0$. A pair of unitary operator $(P_L, P_R)$ is called a $(\beta, b, \varepsilon)$-state-preparation-pair if $P_L \ket{0}_{b} = \sum_{j \in [2^b]} c_j \ket{j}$ and $P_R \ket{0}_b = \sum_{j \in [2^b]} d_j \ket{j}$ such that $\sum_{j \in [m]} \abs{ \beta c_j^* d_j - y_j } \leq \varepsilon$ and $c_j^*d_j = 0$ for all $m \leq j < 2^b$.
    \end{definition}

    \begin{theorem} [Linear combination of block-encoded operators {\cite[Lemma 29]{GSLW19}}] \label{thm:lcu}
        Suppose
        \begin{enumerate}
          \item $y \in \mathbb{C}^m$ with $\Abs{y}_1 \leq \beta$, and $(P_L, P_R)$ is a $(\beta, b, \varepsilon_1)$-state-preparation-pair for $y$.
          \item For every $k \in [m]$, $U_k$ is an $(n+a)$-qubit unitary operator that is an $(\alpha, a, \varepsilon)$-block-encoding of an $n$-qubit operator $A_k$.
        \end{enumerate}
        Then we can implement an $(n+a+b)$-qubit quantum operator $\widetilde U$ using $1$ query to each of $P_L^\dag$, $P_R$ and (controlled-)$U_k$ for $k \in [m]$, and $O(b^2)$ elementary quantum gates such that $\widetilde U$ is an $(\alpha\beta, a+b, \alpha\varepsilon_1+\alpha\beta\varepsilon_2)$-block-encoding of $A = \sum_{k \in [m]} y_k A_k$.
    \end{theorem}
    
    \subsubsection{Product of Block-Encoded Operators}

    The following theorem is a technique to construct a unitary block-encoding of the product of two block-encoded matrices. 
    Suppose we are given two unitary operators that are block-encodings of matrices $A$ and $B$, respectively.
    Then, we are able to implement a block-encoding of the product $AB$ of $A$ and $B$. 
    
    \begin{theorem}
    [Product of block-encoded matrices {\cite[Lemma 30]{GSLW19}}]
    \label{thm:product-block}
    Suppose
    \begin{enumerate}
      \item Unitary operator $U$ is an $(\alpha, a, \delta)$-block-encoding of an $n$-qubit operator $A$.
      \item Unitary operator $V$ is a $(\beta, b, \varepsilon)$-block-encoding of an $n$-qubit operator $B$.
    \end{enumerate}
    Then we can implement a quantum operator $\widetilde U$ using $1$ query to each of $U$ and $V$ such that $\widetilde U$ is an $(\alpha\beta, a+b, \alpha\varepsilon + \beta\delta)$-block-encoding of $AB$.
    \end{theorem}
    
    \subsubsection{Density Operators}

    We describe mixed quantum states as density operators, and introduce how unitary operators prepare purifications of density operators. 

    \begin{definition} [Preparation of density operators] \label{def:subnormalized-density-operator}
        A density operator $\rho$ is a positive semidefinite operator with $\tr(\rho) = 1$. 
        An $(n+a)$-qubit unitary operator $U$ is said to prepare an $n$-qubit density operator $\rho$, if it prepares a purification $\ket\rho = U \ket{0}_{n+a}$ of $\rho$ such that $\rho = \tr_a(\ket{\rho}\bra{\rho})$.
    \end{definition}

    The following theorem shows how to construct a unitary block-encoding of density operators, also known as the technique of purified density matrix \cite{LC19}.
    That is, if we are given a unitary operator that prepares a mixed quantum state $\rho$, then we can implement a block-encoding of $\rho$. 
    
    \begin{theorem} [Block-encoding of density operators, {\cite[Lemma 25]{GSLW19}}] \label{thm:density-to-block-encoding}
        Suppose $U$ is an $\rbra*{n+a}$-qubit unitary operator that prepares
        a purification of an $n$-qubit density operator $\rho$.
        Then, we can implement a $\rbra*{2n+a}$-qubit unitary operator $\widetilde U$ using $1$ query to each of $U$ and $U^\dag$ such that $\widetilde U$ is a $\rbra*{1, n+a, 0}$-block-encoding of $\rho$. 
    \end{theorem}

    The Hadamard test \cite{AJL09} is often used to estimate the value of $\bra{\psi}U\ket{\psi}$ for unitary operator $U$ and quantum state $\ket{\psi}$. In the following, we will introduce a generalized version of Hadamard test that can estimate the value of $\tr\rbra*{A\rho}$ if $A$ is given as block-encoded in unitary operator $U$ and $\rho$ is mixed quantum states. 
    
    \begin{theorem} [Hadamard test, {\cite[Lemma 9]{GP22}}] \label{thm:hadamard-test}
        Suppose $U$ is an $\rbra*{n+a}$-qubit unitary operator that is a $\rbra*{1, a, 0}$-block-encoding of $A$. We can implement a quantum circuit using $1$ query to $U$ and $O\rbra*{1}$ elementary quantum gates such that it outputs $0$ with probability $\frac{1 + \Real\rbra*{\tr\rbra*{A\rho}}}{2}$ (resp. $\frac{1 + \Imag\rbra*{\tr\rbra*{A\rho}}}{2}$) on input $n$-qubit quantum state $\rho$.
    \end{theorem}
    
    By Theorem \ref{thm:hadamard-test}, we can estimate the value of $\tr\rbra*{A\rho}$ within additive error $\varepsilon$ with probability $1 - \delta$ using $O\rbra*{\frac{\log\rbra*{1/\delta}}{\varepsilon^2}}$ samples of $\rho$ and $O\rbra*{\frac{\log\rbra*{1/\delta}}{\varepsilon^2}}$ queries to $U$.
    
    \subsection{Quantum Singular Value Transformation}
    
    Quantum singular value transformation (QSVT) \cite{GSLW19} is a powerful toolbox of quantum computing. 
    Let $f \colon \mathbb{R} \to \mathbb{C}$ be an odd function, i.e., $f\rbra*{x} = -f\rbra*{-x}$. For every operator $A$ with singular value decomposition $A = W \Sigma V^\dag$, where $W$ and $V$ are unitary operators and $\Sigma$ is diagonal with non-negative eigenvalues, define $f^{\SV}\rbra*{A} = W f\rbra*{\Sigma} V^\dag$ as the singular value transformation. 
    In the following, we introduce a special version of QSVT that we need. 
    
    \begin{theorem} [Singular value transformation, {Lemma 19 of the full version of \cite{GSLW19}}] \label{thm:qsvt}
        Suppose 
        \begin{enumerate}
            \item $p \in \mathbb{R}\sbra*{x}$ is an odd polynomial of degree $d$ with $\Abs*{p\rbra*{x}}_{\sbra*{-1, 1}} \leq 1$. 
            \item Unitary operator $U$ is a $\rbra*{1, a, 0}$-block-encoding of operator $A$.
        \end{enumerate}
        Then, we can implement a unitary operator $\widetilde U$ using $\gamma d = O\rbra*{d}$ queries to $U$ for some constant $\gamma > 0$ and $O\rbra*{ad}$ elementary quantum gates such that $\widetilde U$ is a $\rbra*{1, O\rbra*{a}, 0}$-block-encoding of $p^{\SV}\rbra*{A}$. 
    \end{theorem}
    
    Using QSVT, we can approximately perform the sign function.
    This is achieved by the polynomial approximation of the sign function, stated as follows. 
    
    \begin{theorem} [Approximation of the sign function, {\cite[Lemma 14]{GSLW19}}] \label{thm:sgn}
        For $\delta > 0$ and $\varepsilon \in \rbra*{0, 1/2}$, there is an odd polynomial $p \in \mathbb{R}\sbra*{x}$ of degree $d \leq \frac{\eta\log\rbra*{1/\varepsilon}}{\delta}$ for some constant $\eta > 0$ such that
        \begin{enumerate}
            \item $\abs*{p\rbra*{x}} \leq 1$ for all $x \in \sbra*{-2, 2}$.
            \item $\abs*{p\rbra*{x} - \sgn\rbra*{x}} \leq \varepsilon$ for all $x \in \sbra*{-2, 2} \setminus \rbra*{-\delta, \delta}$.
        \end{enumerate}
    \end{theorem}
    
    \subsection{Sampling to Block-Encoding}
    
    Let $\mathcal{D}\rbra*{\mathcal{H}}$ denote the set of density operators on Hilbert space $\mathcal{H}$. For every quantum operator $A$ on Hilbert space $\mathcal{H}$, we define the trace norm of $A$ as $\Abs*{A}_{\tr} = \tr\rbra*{\sqrt{A^\dag A}}$.
    Let $\mathcal{E} \colon \mathcal{D}\rbra*{\mathcal{H}_1} \to \mathcal{D}\rbra*{\mathcal{H}_2}$ be a super-operator (i.e., quantum channel) from Hilbert space $\mathcal{H}_1$ to $\mathcal{H}_2$. The diamond norm of $\mathcal{E}$ is defined by
    \begin{equation}
        \Abs*{\mathcal{E}}_{\diamond} = \max_{\sigma \in \mathcal{D}\rbra*{\mathcal{H}_1^{\otimes 2}} \colon \Abs*{\sigma}_{\tr} \leq 1} \Abs{\rbra*{\mathcal{E} \otimes \mathcal{I}} \rbra*{\sigma}}_{\tr},
    \end{equation}
    where $\mathcal{I} \colon \mathcal{D}\rbra*{\mathcal{H}_1} \to \mathcal{D}\rbra*{\mathcal{H}_1}$ is the identity map on $\mathcal{D}\rbra*{\mathcal{H}_1}$. 

    In this paper, we use quantum channels to approximately implement unitary operators $U$ when dealing with sample access. 
    For this purpose, we introduce the circuit implementations of quantum channels for describing the approximate implementation of $U^\dag$. 
    
    \begin{definition} [Circuit implementations of quantum channels and invertibility]
    \label{def:circuit-impl-qchannel}
        Let $W$ be a quantum unitary circuit and $\rho$ be a mixed quantum state.
        A quantum channel $\mathcal{E}$ is said to be implemented by $\rbra{W, \rho}$, if for every mixed quantum state $\sigma$, $\mathcal{E}\rbra{\sigma} = \tr_{\textup{env}}\rbra{ W \rbra{ {\rho}_{\textup{env}} \otimes \sigma} W^\dag }$.
        We also call the pair $\rbra{W, \rho}$ a circuit implementation of $\mathcal{E}$. 
        
        For $\delta \geq 0$, the circuit implementation $\rbra{W, \rho}$ of a quantum channel $\mathcal{E}$ is said to be $\delta$-invertible with respect to a unitary operator $U$, if $\Abs{\mathcal{E} - U}_\diamond \leq \delta$ and $\Abs{\mathcal{E}^{\textup{inv}} - U^\dag}_\diamond \leq \delta$, where $\mathcal{E}^{\textup{inv}}$ is a quantum channel implemented by $\rbra{W^\dag, \rho}$. 
    \end{definition}

    We show that the composition of invertible circuit implementations of quantum channels is still invertible as follows. 
    \begin{lemma} [Composition of invertible circuit implementations]
    \quad
    \label{lemma:invertible-channel}
    \begin{enumerate}
        \item $\rbra{W \otimes I_{\textup{env}}, \rho_{\textup{env}}}$ is $0$-invertible with respect to $W$ for every unitary operator $W$ and mixed quantum state $\rho$. 
        \item Suppose $\rbra{W_j, \rho_j}$ is $\delta_j$-invertible with respect to $U_j$ for $j = 1, 2$. Then, $\rbra{\rbra{W_1 \otimes I_2} \rbra{W_2 \otimes I_1}, \rho_1 \otimes \rho_2}$ is $\rbra{\delta_1+\delta_2}$-invertible with respect to $U_1U_2$, where $I_j$ is the identity operator on the subspace of $\rho_j$.
    \end{enumerate}
    \end{lemma}
    \begin{proof}
        Item 1) is trivial. To see item 2), suppose that for $j = 1, 2$, $\rbra{W_j, \rho_j}$ implements a a quantum channel $\mathcal{E}_j$ with $\Abs{\mathcal{E}_j - U_j}_\diamond \leq \delta_j$ and $\Abs{\mathcal{E}_j^{\textup{inv}} - U_j^\dag}_\diamond \leq \delta_j$.
        Note that $\rbra{\rbra{W_1 \otimes I_2} \rbra{W_2 \otimes I_1}, \rho_1 \otimes \rho_2}$ implements $\mathcal{E}_1 \circ \mathcal{E}_2$, and $\rbra{\rbra{W_2^\dag \otimes I_1} \rbra{W_1^\dag \otimes I_2}, \rho_1 \otimes \rho_2}$ implements $\rbra{\mathcal{E}_1 \circ \mathcal{E}_2}^{\textup{inv}} = \mathcal{E}_2^{\textup{inv}} \circ \mathcal{E}_1^{\textup{inv}}$.
        Then, $\Abs{\mathcal{E}_1 \circ \mathcal{E}_2 - U_1 \cdot U_2}_\diamond \leq \Abs{\mathcal{E}_1 - U_1}_\diamond + \Abs{\mathcal{E}_2 - U_2}_\diamond \leq \delta_1+\delta_2$ and $\Abs{\rbra{\mathcal{E}_1 \circ \mathcal{E}_2}^{\textup{inv}} - \rbra{U_1 \cdot U_2}^\dag}_\diamond = \Abs{\mathcal{E}_2^{\textup{inv}} \circ \mathcal{E}_1^{\textup{inv}} - U_2^\dag \cdot U_1^\dag}_\diamond \leq \Abs{\mathcal{E}_2^{\textup{inv}} - U_2^\dag}_\diamond + \Abs{\mathcal{E}_1^{\textup{inv}} - U_1^\dag}_\diamond \leq \delta_2+\delta_1$. 
        Therefore, $\rbra{\rbra{W_1 \otimes I_2} \rbra{W_2 \otimes I_1}, \rho_1 \otimes \rho_2}$ is $\rbra{\delta_1+\delta_2}$-invertible with respect to $U_1U_2$. 
    \end{proof}

    The following is the technique of density matrix exponentiation, also known as the sample-based Hamiltonian simulation. 

    \begin{theorem} [Density matrix exponentiation, adapted from \cite{LMR14,KLL+17}]
    \label{thm:den-mat-exp}
    There exists a circuit implementation $\rbra{W, \rho^{\otimes k}}$ that is $\delta$-invertible with respect to (controlled-)$e^{-i\rho t}$, where $k = O\rbra{t^2/\delta}$. 
    \end{theorem}

    In order to modify our quantum algorithm with purified access, we need to construct unitary block-encodings by identical copies of quantum states. 
    This can be done by the technique developed in \cite{GP22} based on density matrix exponentiation \cite{LMR14,KLL+17}. 
    
    \begin{theorem} [Sampling to block-encoding {\cite[Corollary 21]{GP22}}] \label{thm:sample-to-block-encoding}
        Let $\rho$ be an $n$-qubit mixed quantum state.
        There exists a quantum channel $\mathcal{E}$ implemented by $\rbra{W, \rho^{\otimes k} \otimes \ket{0}_\ell\bra{0}}$ such that 
        $\mathcal{E}$ is $\delta$-close to (controlled-)$U$ in the diamond norm,
        where $U$ is a $\rbra*{4/\pi, 3, 0}$-block-encoding of $\rho$, $k = O\rbra*{\frac{\rbra*{\log\rbra*{1/\delta}}^2}{\delta}}$ and $\ell = O\rbra{1}$.
        Moreover, $W$ consists of $O\rbra*{n \cdot \frac{\rbra*{\log\rbra*{1/\delta}}^2}{\delta}}$ elementary quantum gates.
    \end{theorem}

    \begin{lemma}
    \label{lemma:block-encoding-invertible}
        In Theorem \ref{thm:sample-to-block-encoding}, $\rbra{W, \rho^{\otimes k} \otimes \ket{0}_\ell\bra{0}}$ is $\delta$-invertible with respect to $U$.
    \end{lemma}
    \begin{proof} 
        According to \cite{GP22}, $\rbra{W, \rho^{\otimes k} \otimes \ket{0}_\ell\bra{0}}$ is constructed by QSVT techniques and uses queries to $e^{i\rho}$ (as well as its inverse and controlled versions), where $e^{i\rho}$ is approximately implemented by Theorem \ref{thm:den-mat-exp}. 
        As $\rbra{W, \rho^{\otimes k} \otimes \ket{0}_\ell\bra{0}}$ is only composed of unitary operators (by QSVT) and invertible circuit implementations of $e^{i\rho}$ (by density matrix exponentiation), it is therefore $\delta$-invertible with respect to $U$ by Lemma \ref{lemma:invertible-channel}. 
    \end{proof}
    
    \section{The Algorithm} \label{sec:algo}
    
    In this section, we will first provide a quantum algorithm for low-rank trace distance estimation with purified access; and then modify it to another algorithm with sample access. 
    The algorithms will be written in a general form (see Theorem \ref{thm:purified} and Theorem \ref{thm:sample}) using the notions introduced for approximate rank in Section \ref{sec:approx-rank}, and low-rank trace distance estimation will be considered to be their corollaries (see Corollary \ref{corollary:purified} and Corollary \ref{corollary:sample}).
    
    \subsection{Purified Access}
    
    In the purified quantum query access model, mixed quantum state $\rho$ is given by a unitary operator $O_\rho$ that prepares its purification. That is, 
    \begin{equation}
        O_\rho \ket{0}_{n+n_\rho} = \ket{\rho}_{n+n_\rho}, 
    \end{equation}
    \begin{equation}
        \rho = \tr_{n_\rho} \rbra*{\ket{\rho}_{n+n_\rho} \bra{\rho}},
    \end{equation}
    where $n_\rho$ is the number of ancilla qubits and we usually assume that $n_\rho \leq n$. 
    Before we state the main theorem, let us recall that $w\rbra{A, \delta}$ denotes the sum of absolute eigenvalues of $A$ that are not greater than $\delta$ (see Eq. (\ref{eq:w})).
    \begin{theorem} \label{thm:purified}
        Given quantum oracles $O_\rho$ and $O_\sigma$ that prepare $N$-dimensional quantum states $\rho$ and $\sigma$, respectively, for every $\delta_p > 0$ such that 
        \begin{equation}
        w\rbra*{\frac{\rho - \sigma}{2}, \delta_p} \leq \frac{\varepsilon}{4}, 
        \end{equation}
        with $w\rbra{\cdot, \cdot}$ defined in Eq. (\ref{eq:w}),
        there is a quantum algorithm that computes the trace distance $T\rbra*{\rho, \sigma}$ within additive error $\varepsilon$ using 
        \begin{equation}
        O\rbra*{\frac{1}{\delta_p\varepsilon} \log\rbra*{\frac{1}{\varepsilon}}}    
        \end{equation}
        queries to these oracles 
        and 
        \begin{equation}
        O\rbra*{\frac{1}{\delta_p\varepsilon} \log\rbra*{\frac{1}{\varepsilon}}\log\rbra*{N}}
        \end{equation}
        elementary quantum gates. 
    \end{theorem}
    
    \begin{proof}
    
    Let $\nu = \rbra{\rho - \sigma} / 2$ with singular value decomposition $\nu = W \Sigma V^\dag$. Then, 
    \begin{align}
        T\rbra*{\rho, \sigma}
        & = \tr\rbra*{\abs*{\frac{\rho-\sigma}{2}}} \\
        & = \Abs*{\nu}_{\tr} \\
        & = \tr \rbra*{ \nu \sgn^{\SV}\rbra*{\nu} } \\
        & = \frac 1 2 \rbra*{ \tr\rbra*{ \sgn^{\SV}\rbra*{\nu} \rho } - \tr\rbra*{ \sgn^{\SV}\rbra*{\nu} \sigma } }.
    \end{align}
    The main idea of our algorithm is to estimate $x_{\rho} \approx \tr\rbra{\sgn^{\SV}\rbra*{\nu} \rho}$ and $x_{\sigma} \approx \tr\rbra{\sgn^{\SV}\rbra*{\nu} \sigma}$, and then output $\rbra*{x_{\rho} - x_{\sigma}} / 2$ as the estimate of the trace distance $T\rbra*{\rho, \sigma}$.

    \textbf{Step 1: Implement the block-encodings of $\rho$ and $\sigma$.}
    Suppose $O_\rho$ and $O_\sigma$ are $\rbra*{n+n_\rho}$-qubit and $\rbra*{n+n_\sigma}$-qubit quantum unitary oracles that prepare $n$-qubit mixed quantum states $\rho$ and $\sigma$, respectively, where $N = 2^n$ and $\max\cbra*{n_\rho, n_\sigma} \leq n$. By Theorem \ref{thm:density-to-block-encoding}, we can obtain unitary operators $U_\rho$ and $U_\sigma$ using $O\rbra*{1}$ queries to $O_\rho$ and $O_\sigma$, respectively, such that $U_\rho$ is a $\rbra*{1, n+n_\rho, 0}$-block-encoding of $\rho$ and $U_\sigma$ is a $\rbra*{1, n+n_\sigma, 0}$-block-encoding of $\sigma$. 

    \textbf{Step 2: Implement the block-encoding of $\nu = \rbra{\rho-\sigma}/2$.}
    According to Definition \ref{def:state-preparation-pair}, we note that $(HX, H)$ is a $(2, 1, 0)$-state-preparation-pair for $y = (1, -1)$, where $H$ is the Hadamard gate and $X$ is the Pauli matrix.
    By Theorem \ref{thm:lcu}, there is a quantum operator $U_{\nu}$ using $1$ query to each of $U_\rho$ and $U_\sigma$ and $O(1)$ elementary quantum gates such that $U_\nu$ is a $(1, O\rbra*{n+n_\rho+n_\sigma}, 0)$-block-encoding of $\nu = (\rho - \sigma)/2$.

    \textbf{Step 3: Implement the block-encoding of $\sgn^{\SV}\rbra{\nu}$.}
    Now we start from $U_\nu$, a $\rbra*{1, O\rbra*{n+n_\rho+n_\sigma}, 0}$-block-encoding of $\nu$, to construct a block-encoding of $\sgn^{\SV}\rbra*{\nu}$. By Theorem \ref{thm:sgn}, we have an odd polynomial $p \in \mathbb{R}\sbra*{x}$ of degree $d_p = O\rbra*{\frac{\log\rbra*{1/\varepsilon_p}}{\delta_p}}$, where $\varepsilon_p \in \rbra*{0, 1/2}$ is to be determined, such that
    \begin{enumerate}
        \item $\abs*{p\rbra*{x}} \leq 1$ for all $x \in \sbra*{-2, 2}$.
        \item $\abs*{p\rbra*{x} - \sgn\rbra*{x}} \leq \varepsilon_p$ for all $x \in \sbra*{-2, 2} \setminus \rbra*{-\delta_p, \delta_p}$.
    \end{enumerate}
    By Theorem \ref{thm:qsvt}, we can implement a unitary operator $U_{p^{\SV}\rbra*{\nu}}$ using $O\rbra*{d_p}$ queries to $U_\nu$ and $O\rbra*{\rbra*{n+n_\rho+n_\sigma}d_p}$ elementary quantum gates such that $U_{p^{\SV}\rbra*{\nu}}$ is a $\rbra*{1, O\rbra*{n+n_\rho+n_\sigma}, 0}$-block-encoding of $p^{\SV}\rbra*{\nu}$.

    \textbf{Step 4: Estimate $\tr\rbra{p^{\SV}\rbra{\nu} \rho}$ and $\tr\rbra{p^{\SV}\rbra{\nu} \sigma}$}.
    Combining Theorem \ref{thm:hadamard-test} and Theorem \ref{thm:amp-estimation}, we can obtain an estimation $x_\rho$ of $\tr\rbra*{p^{\SV}\rbra*{\nu} \rho}$ within additive error $\varepsilon_H$ with high probability using $O\rbra*{1/\varepsilon_H}$ queries to $U_{p^{\SV}\rbra*{\nu}}$ and $O_\rho$. 
    Similarly, we can obtain an estimation $x_\sigma$ of $\tr\rbra*{p^{\SV}\rbra*{\nu} \sigma}$ within additive error $\varepsilon_H$ with high probability using $O\rbra*{1/\varepsilon_H}$ queries to $U_{p^{\SV}\rbra*{\nu}}$ and $O_\sigma$. 
    That is, 
    \begin{equation}
        \abs*{x_\rho - \tr\rbra*{p^{\SV}\rbra*{\nu} \rho} } \leq \varepsilon_H,
    \end{equation}
    \begin{equation}
        \abs*{x_\sigma - \tr\rbra*{p^{\SV}\rbra*{\nu} \sigma} } \leq \varepsilon_H.
    \end{equation}

    \textbf{Step 5: Estimate the trace distance}.
    Finally, we output $\rbra*{x_\rho - x_\sigma} / 2$ as the estimate of $T\rbra*{\rho, \sigma}$. 
    
    \textbf{Error analysis}.
    Let $\nu = \sum_{j \in \sbra*{N}} \lambda_j \ket{\psi_j}\bra{\psi_j}$ be the spectral decomposition of $\nu$. Since $\nu$ is Hermitian, we have $p^{\SV}\rbra*{\nu} = p\rbra*{\nu}$ and $\sgn^{\SV}\rbra*{\nu} = \sgn\rbra*{\nu}$. Moreover, 
    \begin{align}
        & \abs*{ \tr\rbra*{\nu p^{\SV}\rbra*{\nu}} - \tr\rbra*{\nu \sgn^{\SV}\rbra*{\nu}} } \nonumber \\
        & \qquad \leq \sum_{j \in \sbra*{N}} \abs*{ \lambda_j p\rbra*{\lambda_j} - \lambda_j } \\
        & \qquad = \sum_{\abs*{\lambda_j} > \delta_p} \abs*{\lambda_j} \abs*{ p\rbra*{\lambda_j} - 1 } + \sum_{\abs*{\lambda_j} \leq \delta_p} \abs*{\lambda_j} \abs*{ p\rbra*{\lambda_j} - 1 } \\
        & \qquad \leq \sum_{\abs*{\lambda_j} > \delta_p} \abs*{\lambda_j} \varepsilon_p + \sum_{\abs*{\lambda_j} \leq \delta_p} 2 \abs*{\lambda_j} \\
        & \qquad \leq 2 \varepsilon_p + 2 w\rbra*{\nu, \delta_p} \\
        & \qquad \leq 2 \varepsilon_p + \frac{\varepsilon} {2}.
    \end{align}
    Therefore, with probability $O\rbra*{1}$, we have
    \begin{align}
        \abs*{ \frac{x_\rho - x_\sigma}{2} - T\rbra*{\rho, \sigma} } &
        \leq \frac 1 2 \abs*{x_\rho - \tr\rbra*{ \sgn^{\SV}\rbra*{\nu} \rho }} + \frac 1 2 \abs*{x_\sigma - \tr\rbra*{ \sgn^{\SV}\rbra*{\nu} \sigma }} \nonumber \\ & \qquad + 
        \abs*{ \tr\rbra*{\nu p^{\SV}\rbra*{\nu}} - \tr\rbra*{\nu \sgn^{\SV}\rbra*{\nu}} } \\
        & \leq \varepsilon_H + 2\varepsilon_p + \frac{\varepsilon}{2}.
    \end{align}

    \textbf{Complexity analysis}. 
    By letting $\varepsilon_p = \varepsilon/8$ and $\varepsilon_H = \varepsilon/4$, the query complexity is
    \begin{equation}
        O\rbra*{\frac{\log\rbra*{1/\varepsilon_p}}{\delta_p} \cdot \frac{1}{\varepsilon_H}} = O\rbra*{\frac{1}{\delta_p\varepsilon} \log\rbra*{\frac{1}{\varepsilon}}}.
    \end{equation}
    Furthermore, the number of elementary quantum gates is 
    \begin{equation}
    O\rbra*{\frac{1}{\delta_p\varepsilon} \log\rbra*{\frac{1}{\varepsilon}} \log\rbra*{N}}.\qedhere
    \end{equation}
    \end{proof}
    
    See Algorithm \ref{algo:purified} for a formal description of our algorithm in Theorem \ref{thm:purified}.
    
    \begin{algorithm}[t]
        \caption{Quantum algorithm for trace distance estimation given purified access.}
        \label{algo:purified}
        \begin{algorithmic}[1]
        \Require Quantum oracles $O_\rho$ and $O_\sigma$ that prepare mixed quantum states $\rho$ and $\sigma$, respectively; the desired additive error $\varepsilon > 0$; and $\delta_p > 0$ such that $w\rbra*{\rbra*{\rho - \sigma}/2, \delta_p} \leq \varepsilon / 4$.
        \Ensure An estimate of $T(\rho, \sigma)$ within additive error $\varepsilon$ with probability $O\rbra*{1}$. 
        
        \State $\varepsilon_p \gets \varepsilon/8$.
        
        \State $\varepsilon_H \gets \varepsilon/4$.
        
        \State $U_\rho$ and $U_\sigma$, unitary operators using $O\rbra*{1}$ queries to $O_\rho$ and $O_\sigma$ (by Theorem \ref{thm:density-to-block-encoding}), are $\rbra*{1, O\rbra*{n}, 0}$-block-encodings of $\rho$ and $\sigma$, respectively. 
        
        \State $U_\nu$, a unitary operator using $1$ query to each of $U_\rho$ and $U_\sigma$ (by Theorem \ref{thm:lcu}), is a $\rbra*{1, O\rbra*{n}, 0}$-block-encoding of $\nu = \rbra*{\rho-\sigma}/2$.
        
        \State Let $p \in \mathbb{R}\sbra*{x}$ be an odd polynomial of degree $d_p = O\rbra*{\frac{\log\rbra*{1/\varepsilon_p}}{\delta_p}}$ (by Theorem \ref{thm:sgn}) such that 
        \begin{enumerate}
            \item $\abs*{p\rbra*{x}} \leq 1$ for all $x \in \sbra*{-2, 2}$.
            \item $\abs*{p\rbra*{x} - \sgn\rbra*{x}} \leq \varepsilon_p$ for all $x \in \sbra*{-2, 2} \setminus \rbra*{-\delta_p, \delta_p}$.
        \end{enumerate}
        
        \State $U_{p^{\SV}\rbra*{\nu}}$, a unitary operator using $O\rbra*{d_p}$ queries to $U_\nu$ (by Theorem \ref{thm:qsvt}), is a $\rbra*{1, O\rbra*{n}, 0}$-block-encoding of $p^{\SV}\rbra*{\nu}$.
        
        \State $x_\rho \gets \tr\rbra*{p^{\SV}\rbra*{\nu} \rho} \pm \varepsilon_H$ with probability $O\rbra*{1}$ using $O\rbra*{1/\varepsilon_H}$ queries to $U_{p^{\SV}\rbra*{\nu}}$ and $O_\rho$ (by Theorem \ref{thm:hadamard-test} and Theorem \ref{thm:amp-estimation}).
        
        \State $x_\sigma \gets \tr\rbra*{p^{\SV}\rbra*{\nu} \sigma} \pm \varepsilon_H$ with probability $O\rbra*{1}$ using $O\rbra*{1/\varepsilon_H}$ queries to $U_{p^{\SV}\rbra*{\nu}}$ and $O_\sigma$ (by Theorem \ref{thm:hadamard-test} and Theorem \ref{thm:amp-estimation}).
        
        \State \Return $\rbra*{x_\rho - x_\sigma} / 2$.
        \end{algorithmic}
    \end{algorithm}
    
    \begin{corollary} [Low-rank trace distance estimation with purified access] \label{corollary:purified}
        Given quantum oracles $O_\rho$  and $O_\sigma$ that prepare $N$-dimensional quantum states $\rho$ and $\sigma$, respectively, 
        there is a quantum algorithm that computes the trace distance $T\rbra*{\rho, \sigma}$ within additive error $\varepsilon$ using 
        \begin{equation}
            O\rbra*{\frac{r}{\varepsilon^2} \log\rbra*{\frac{1}{\varepsilon}}}
        \end{equation}
        queries to these oracles 
        and 
        \begin{equation}
        O\rbra*{\frac{r}{\varepsilon^2} \log\rbra*{\frac{1}{\varepsilon}}\log\rbra*{N}}
        \end{equation}
        elementary quantum gates, where $r$ is the upper bound of the rank of $\rho$ and $\sigma$. 
    \end{corollary}
    \begin{proof}
        Taking $\delta_p = \varepsilon/8r$ in Theorem \ref{thm:purified}, we will obtain the desired complexity by noting that
        \begin{align}
            w\rbra*{\frac{\rho - \sigma}{2}, \delta_p} 
            & \leq \delta_p \cdot \rank\rbra*{\frac{\rho-\sigma}{2}} \\
            & \leq \frac{\varepsilon}{8r} \cdot \rbra*{\rank\rbra*{\rho} + \rank\rbra*{\sigma}} \\
            & \leq \frac{\varepsilon}{8r} \cdot 2r 
            = \frac{\varepsilon}{4}.
        \end{align}
    \end{proof}

    \subsection{Sample Access}
    
    In this subsection, we will provide a quantum algorithm given sample access.
    In this algorithm, our idea is to use quantum channels to approximately implement the unitary operators $U_\rho$ and $U_\sigma$ used in Algorithm \ref{algo:purified}. 
    By Theorem \ref{thm:sample-to-block-encoding}, we can implement quantum channels for (controlled-)$U_\rho$ and (controlled-)$U_\sigma$, using only samples of $\rho$ and $\sigma$.
    Also, quantum channels for (controlled-)$U_\rho^\dag$ and (controlled-)$U_\sigma^\dag$ can also be implemented similarly by Lemma \ref{lemma:block-encoding-invertible}.

    Now we are ready to show our quantum algorithm for estimating the trace distance between two mixed quantum states given sample access as follows.
    
    \begin{theorem} \label{thm:sample}
        Given access to identical copies of $N$-dimensional quantum states $\rho$ and $\sigma$, for every $\delta_p > 0$ such that 
        \begin{equation}
            w\rbra*{\frac{\rho - \sigma}{2}, \delta_p} \leq \frac{\varepsilon}{4}, 
        \end{equation}
        with $w\rbra{\cdot, \cdot}$ defined in Eq. (\ref{eq:w}),
        there is a quantum algorithm that computes the trace distance $T\rbra*{\rho, \sigma}$ within additive error $\varepsilon$ using 
        \begin{equation}
        O\rbra*{\frac{1}{\delta_p^2\varepsilon^3}
        \log^2\rbra*{\frac{1}{\delta_p\varepsilon}} \log^2\rbra*{\frac{1}{\varepsilon}}}
        \end{equation}
        samples of $\rho$ and $\sigma$ and 
        \begin{equation}
        O\rbra*{\frac{1}{\delta_p^2\varepsilon^3} \log^2\rbra*{\frac{1}{\delta_p\varepsilon}} \log^2\rbra*{\frac{1}{\varepsilon}} \log\rbra*{N}}
        \end{equation}
        elementary quantum gates. In addition, the depth of the quantum circuit is 
        \begin{equation}
        O\rbra*{\frac{1}{\delta_p^2\varepsilon} \log^2\rbra*{\frac{1}{\delta_p\varepsilon}} \log^2\rbra*{\frac{1}{\varepsilon}} \log\rbra*{N}}.
        \end{equation}
    \end{theorem}
    
    \begin{proof}
    
    The algorithm follows but is more complicated than that in Theorem \ref{thm:purified}. 
    
    \textbf{Step 1: Implement the block-encodings of $\rho$ and $\sigma$.}
    The first step is to convert samples of $\rho$ and $\sigma$ to their block-encodings, respectively. 
    By Theorem \ref{thm:sample-to-block-encoding}, we can implement a quantum super-operator $\mathcal{E}_\rho$ using  $O\rbra*{\frac{\rbra*{\log\rbra*{1/\delta}}^2}{\delta}}$ samples of $\rho$ and $O\rbra*{n \cdot \frac{\rbra*{\log\rbra*{1/\delta}}^2}{\delta}}$ elementary quantum gates such that $\Abs*{ \mathcal{E}_\rho - \mathcal{U}_\rho }_{\diamond} \leq \delta$,
    where $\mathcal{U}_\rho\rbra*{\cdot} = U_\rho \rbra*{\cdot} U_\rho^\dag$ and unitary operator $U_\rho$ is a $\rbra*{4/\pi, 3, 0}$-block-encoding of $\rho$. We can also obtain a quantum super-operator $\mathcal{E}_\sigma$ for unitary operator $U_\sigma$ (block-encoding of $\sigma$) similar to that for $\rho$. In the following, quantum super-operators will be used as if they were unitary operators.

    \textbf{Step 2: Implement the block-encoding of $\nu = \rbra{\rho-\sigma}/2$.}
    According to Definition \ref{def:state-preparation-pair}, we note that $(HX, H)$ is a $(2, 1, 0)$-state-preparation-pair for $y = (1, -1)$, where $H$ is the Hadamard gate and $X$ is the Pauli matrix.
    By Theorem \ref{thm:lcu}, there is a quantum super-operator $\mathcal{E}_{\nu}$ using $1$ query to each of $\mathcal{E}_\rho$ and $\mathcal{E}_\sigma$ and $O(1)$ elementary quantum gates such that $\mathcal{E}_\nu$ is $2\delta$-close in the diamond norm to a $(4/\pi, O\rbra*{1}, 0)$-block-encoding of $\nu = (\rho - \sigma)/2$.

    \textbf{Step 3: Implement the block-encoding of $\sgn^{\SV}\rbra{\nu}$.}
    Now we start from $\mathcal{E}_\nu$, a $2\delta$-close in the diamond norm quantum super-operator to a $\rbra*{4/\pi, O\rbra*{1}, 0}$-block-encoding of $\nu$, to construct a quantum super-operator close to a block-encoding of $\sgn^{\SV}\rbra*{\nu}$. By Theorem \ref{thm:sgn}, we have an odd polynomial $p \in \mathbb{R}\sbra*{x}$ of degree $d_p \leq  \frac{\eta\log\rbra*{1/\varepsilon_p}}{\delta_p}$ for some constant $\eta > 0$, where $\varepsilon_p \in \rbra*{0, 1/2}$ is to be determined, such that
    \begin{enumerate}
        \item $\abs*{p\rbra*{x}} \leq 1$ for all $x \in \sbra*{-2, 2}$.
        \item $\abs*{p\rbra*{x} - \sgn\rbra*{x}} \leq \varepsilon_p$ for all $x \in \sbra*{-2, 2} \setminus \rbra*{-\delta_p, \delta_p}$.
    \end{enumerate}
    By Theorem \ref{thm:qsvt}, we can implement a quantum super-operator $\mathcal{E}_{p^{\SV}\rbra*{\nu}}$ using $q \leq \gamma d_p = O\rbra*{d_p}$ queries to $\mathcal{E}_\nu$ for some constant $\gamma > 0$ and $O\rbra*{d_p}$ elementary quantum gates such that $\mathcal{E}_{p^{\SV}\rbra*{\nu}}$ is $2q\delta$-close in the diamond norm to a $\rbra*{4/\pi, O\rbra*{1}, 0}$-block-encoding of $p^{\SV}\rbra*{\nu}$.

    \textbf{Step 4: Estimate $\tr\rbra{p^{\SV}\rbra{\nu} \rho}$ and $\tr\rbra{p^{\SV}\rbra{\nu} \sigma}$}.
    By Theorem \ref{thm:hadamard-test}, we can obtain an estimation $\widetilde x_\rho$ of $\tr\rbra*{p^{\SV}\rbra*{\nu} \rho}$ within additive error $\varepsilon_H + 2q\delta$ with probability $O\rbra*{1}$ using $O\rbra*{1/\varepsilon_H^2}$ repetitions of Hadamard test, where each repetition uses $1$ query to $\mathcal{E}_{p^{\SV}\rbra*{\nu}}$. That is, with probability $O\rbra*{1}$, we have
    \begin{equation}
        \abs*{\widetilde x_\rho - \frac{\pi}{4} \tr\rbra*{p^{\SV}\rbra*{\nu} \rho} } \leq \varepsilon_H + 2q\delta,
    \end{equation}
    where $\varepsilon_H$ is from the Hadamard test, and $2q\delta$ is due to the error in the diamond norm.
    Similarly, we can obtain an estimation $\widetilde x_\sigma$ within additive error $\varepsilon_H + 2q\delta$ with probability $O\rbra*{1}$ such that
    \begin{equation}
        \abs*{\widetilde x_\sigma - \frac{\pi}{4}\tr\rbra*{p^{\SV}\rbra*{\nu} \sigma} } \leq \varepsilon_H + 2q\delta.
    \end{equation}

    \textbf{Step 5: Estimate the trace distance}.
    Finally, we output $2\rbra*{\widetilde x_\rho - \widetilde x_\sigma} / \pi$ as the estimate of $T\rbra*{\rho, \sigma}$. 
    
    \textbf{Error analysis}. 
    Combining the above, with probability $O\rbra*{1}$, we have
    \begin{align}
        & \abs*{ \frac{2}{\pi}\rbra*{\widetilde x_\rho - \widetilde x_\sigma} - T\rbra*{\rho, \sigma} } \nonumber \\
        & \leq \abs*{ \frac{2}{\pi}\rbra*{\widetilde x_\rho - \widetilde x_\sigma} - \frac{\tr\rbra*{ p^{\SV}\rbra*{\nu} \rho} - \tr\rbra*{ p^{\SV}\rbra*{\nu}\sigma}}{2} } + \abs*{ \frac{\tr\rbra*{ p^{\SV}\rbra*{\nu} \rho} - \tr\rbra*{ p^{\SV}\rbra*{\nu}\sigma}}{2} - T\rbra*{\rho, \sigma} } \\
        & \leq \frac{4}{\pi} \rbra*{ \varepsilon_H + 2q\delta } + 2\varepsilon_p + \frac{\varepsilon}{2} \\
        & \leq \frac{8 \gamma \eta \delta \log\rbra*{1/\varepsilon_p}}{\pi\delta_p} + \frac{4\varepsilon_H}{\pi} + 2\varepsilon_p + \frac{\varepsilon}{2}. 
    \end{align}
    
    \textbf{Complexity analysis}. 
    By letting $\varepsilon_p = \varepsilon/12$, $\varepsilon_H = \pi\varepsilon/24$, and $\delta = \frac{\pi \varepsilon \delta_p}{48 \gamma \eta \log\rbra*{1/\varepsilon_p}}$, the sample complexity is
    \begin{align}
        O\rbra*{\frac{\rbra*{\log\rbra*{1/\delta}}^2}{\delta} \cdot \frac{\log\rbra*{1/\varepsilon_p}}{\delta_p} \cdot \frac{1}{\varepsilon_H^2}} 
        = O\rbra*{\frac{1}{\delta_p^2\varepsilon^3}
        \log^2\rbra*{\frac{1}{\delta_p\varepsilon}} \log^2\rbra*{\frac{1}{\varepsilon}}}.
    \end{align}
    Furthermore, the number of elementary quantum gates is 
    \begin{equation}
    O\rbra*{\frac{1}{\delta_p^2\varepsilon^3} \log^2\rbra*{\frac{1}{\delta_p\varepsilon}} \log^2\rbra*{\frac{1}{\varepsilon}} \log\rbra*{N}},
    \end{equation}
    and the depth of the quantum circuit is
    \begin{equation}
    O\rbra*{\frac{1}{\delta_p^2\varepsilon} \log^2\rbra*{\frac{1}{\delta_p\varepsilon}} \log^2\rbra*{\frac{1}{\varepsilon}} \log\rbra*{N}}.
    \end{equation}
    \end{proof}
    
    See Algorithm \ref{algo:sample} for a formal description of our algorithm in Theorem \ref{thm:sample}.
    
    \begin{algorithm}[t]
        \caption{Quantum algorithm for trace distance estimation given sample access.}
        \label{algo:sample}
        \begin{algorithmic}[1]
        \Require Identical copies of quantum states $\rho$ and $\sigma$; the desired additive error $\varepsilon > 0$; and $\delta_p > 0$ such that $w\rbra*{\rbra*{\rho - \sigma}/2, \delta_p} \leq \varepsilon / 4$.
        \Ensure An estimate of $T(\rho, \sigma)$ within additive error $\varepsilon$ with probability $O\rbra*{1}$. 
        
        \State $\varepsilon_p \gets \varepsilon/12$.
        
        \State $\varepsilon_H \gets \pi\varepsilon/24$.
        
        \State $\delta \gets \frac{\pi \varepsilon \delta_p}{48 \gamma \eta \log\rbra*{1/\varepsilon_p}}$, where $\gamma$ and $\eta$ are the constants in Theorem \ref{thm:qsvt} and Theorem \ref{thm:sgn}, respectively.
        
        \State $\mathcal{E}_\rho$ and $\mathcal{E}_\sigma$, quantum super-operators using  $O\rbra*{\frac{\rbra*{\log\rbra*{1/\delta}}^2}{\delta}}$ samples of $\rho$ and $\sigma$, are $\delta$-close in the diamond norm to certain unitary operators that are $\rbra*{4/\pi, 3, 0}$-block-encodings of $\rho$ and $\sigma$, respectively.
        
        \State $\mathcal{E}_\nu$, a quantum super-operator using $1$ query to each of $\mathcal{E}_\rho$ and $\mathcal{E}_\sigma$ (by Theorem \ref{thm:lcu} as if they were unitary operators), is $2\delta$-close in the diamond norm to a $\rbra*{4/\pi, O\rbra*{1}, 0}$-block-encoding of $\nu = \rbra*{\rho-\sigma}/2$.
        
        \State Let $p \in \mathbb{R}\sbra*{x}$ be an odd polynomial of degree $d_p \leq \frac{\eta\log\rbra*{1/\varepsilon_p}}{\delta_p}$ (by Theorem \ref{thm:sgn}) such that 
        \begin{enumerate}
            \item $\abs*{p\rbra*{x}} \leq 1$ for all $x \in \sbra*{-2, 2}$.
            \item $\abs*{p\rbra*{x} - \sgn\rbra*{x}} \leq \varepsilon_p$ for all $x \in \sbra*{-2, 2} \setminus \rbra*{-\delta_p, \delta_p}$.
        \end{enumerate}
        
        \State $\mathcal{E}_{p^{\SV}\rbra*{\nu}}$, a quantum super-operator using $q \leq \gamma d_p$ queries to $\mathcal{E}_\nu$ (by Theorem \ref{thm:qsvt} as if it were a unitary operator), is $2 q \delta$-close in the diamond norm to a $\rbra*{4/\pi, O\rbra*{1}, 0}$-block-encoding of $p^{\SV}\rbra*{\nu}$.
        
        \State $\widetilde x_\rho \gets \frac{\pi}{4} \tr\rbra*{p^{\SV}\rbra*{\nu} \rho} \pm \rbra*{ \varepsilon_H + 2q\delta }$ with probability $O\rbra*{1}$ using $O\rbra*{1/\varepsilon_H^2}$ queries to $\mathcal{E}_{p^{\SV}\rbra*{\nu}}$ (as if it were a unitary operator) and $O\rbra*{1/\varepsilon_H^2}$ samples of $\rho$ (by Theorem \ref{thm:hadamard-test}).
        
        \State $\widetilde x_\sigma \gets \frac{\pi}{4} \tr\rbra*{p^{\SV}\rbra*{\nu} \sigma} \pm \rbra*{ \varepsilon_H + 2q\delta }$ with probability $O\rbra*{1}$ using $O\rbra*{1/\varepsilon_H^2}$ queries to $\mathcal{E}_{p^{\SV}\rbra*{\nu}}$ (as if it were a unitary operator) and $O\rbra*{1/\varepsilon_H^2}$ samples of $\sigma$ (by Theorem \ref{thm:hadamard-test}).
        
        \State \Return $2 \rbra*{\widetilde x_\rho - \widetilde x_\sigma} / \pi$.
        \end{algorithmic}
    \end{algorithm}
    
    \begin{corollary}[Low-rank trace distance estimation with sample access] \label{corollary:sample}
        Given access to identical copies of $N$-dimensional quantum states $\rho$ and $\sigma$, there is a quantum algorithm that computes the trace distance $T\rbra*{\rho, \sigma}$ within additive error $\varepsilon$ using
        \begin{equation}
            O\rbra*{\frac{r^2}{\varepsilon^5} \log^2\rbra*{\frac{r}{\varepsilon}} \log^2\rbra*{\frac{1}{\varepsilon}}}
        \end{equation}
        samples of $\rho$ and $\sigma$ and 
        \begin{equation}
            O\rbra*{\frac{r^2}{\varepsilon^5} \log^2\rbra*{\frac{1}{\delta_p\varepsilon}} \log^2\rbra*{\frac{1}{\varepsilon}} \log\rbra*{N}}
        \end{equation}
        elementary quantum gates, where $r$ is the upper bound of the rank of $\rho$ and $\sigma$. In addition, the depth of the quantum circuit is 
        \begin{equation}
            O\rbra*{\frac{r^2}{\varepsilon^3} \log^2\rbra*{\frac{r}{\varepsilon}} \log^2\rbra*{\frac{1}{\varepsilon}} \log\rbra*{N}}.
        \end{equation}
    \end{corollary}
    \begin{proof}
    Taking $\delta_p = \varepsilon/8r$ in Theorem \ref{thm:sample}, we will obtain the desired complexity by noting that
        \begin{align}
            w\rbra*{\frac{\rho - \sigma}{2}, \delta_p} 
            & \leq \delta_p \cdot \rank\rbra*{\frac{\rho-\sigma}{2}} \\
            & \leq \frac{\varepsilon}{8r} \cdot \rbra*{\rank\rbra*{\rho} + \rank\rbra*{\sigma}} \\
            & \leq \frac{\varepsilon}{8r} \cdot 2r 
            = \frac{\varepsilon}{4}.
        \end{align}
    \end{proof}

    \section{BQP-Completeness}
    \label{sec:bqp-completeness}

    In this section, we show the $\mathsf{BQP}$-completeness of (the decision version of) low-rank trace distance estimation, which is reduced from fidelity estimation of pure quantum states. 

    \begin{definition} [Low-rank trace distance testing]
        Let $n \geq 1$ be a positive integer.
        For $0 < \alpha < \beta < 1$, let $\textsc{TrDisTest}\rbra{n, \alpha, \beta}$ be a decision problem defined as follows. 
        Given the classical description of $\poly\rbra{n}$-size quantum circuits $O_{\rho}$ and $O_{\sigma}$ that prepare purifications of $n$-qubit mixed quantum states $\rho$ and $\sigma$ of rank $\poly\rbra{n}$, determine whether $T\rbra{\rho, \sigma} < \alpha$ or $T\rbra{\rho, \sigma} > \beta$, promised that it is in either case. 
    \end{definition}

    \begin{theorem}
    \label{thm:td-est-bqp}
        $\textsc{TrDisTest}\rbra{n, \alpha, \beta}$ is $\mathsf{BQP}$-complete for every $2^{-\poly\rbra{n}} \leq \alpha < \beta \leq 1 - 2^{-\poly\rbra{n}}$ with $\beta - \alpha \geq 1/\poly\rbra{n}$, even if both of the quantum states are pure. 
    \end{theorem}

        To show the $\mathsf{BQP}$-hardness of $\textsc{TrDisTest}\rbra{n, \alpha, \beta}$, we need the result in \cite{RASW23} regarding pure-state fidelity.
        \begin{definition} [Pure-state fidelity testing]
            Let $n \geq 1$ be a positive integer. 
            For every $0 < \alpha < \beta < 1$, let $\textsc{PureFidTest}\rbra{n, \alpha, \beta}$ be a decision problem defined as follows. 
            Given the classical description of $\poly\rbra{n}$-size quantum circuits $O_{\psi}$ and $O_{\phi}$ that prepare $n$-qubit pure quantum states $\ket{\psi}$ and $\ket{\phi}$, determine whether $\abs{\braket{\psi}{\phi}}^2 < \alpha$ or $\abs{\braket{\psi}{\phi}}^2 > \beta$, promised that it is in either case. 
        \end{definition}
        \begin{theorem} [Pure-state fidelity testing is $\mathsf{BQP}$-complete, {\cite[Theorem 12]{RASW23}}]
        \label{thm:fidtest}
            $\textsc{PureFidTest}\rbra{n, \alpha, \beta}$ is $\mathsf{BQP}$-complete for every $2^{-\poly\rbra{n}} \leq \alpha < \beta \leq 1 - 2^{-\poly\rbra{n}}$ with $\beta - \alpha \geq 1/\poly\rbra{n}$.
        \end{theorem}

    \begin{proof} [Proof of Theorem \ref{thm:td-est-bqp}]
        Directly applying the polynomial-time quantum algorithm for low-rank trace distance estimation given in Theorem \ref{thm:purified}, we have that $\textsc{TrDisTest}\rbra{n, \alpha, \beta}$ is in $\mathsf{BQP}$ for $2^{-\poly\rbra{n}} \leq \alpha < \beta \leq 1 - 2^{-\poly\rbra{n}}$ with $\beta - \alpha \geq 1/\poly\rbra{n}$.
        It remains to show the $\mathsf{BQP}$-hardness of $\textsc{TrDisTest}\rbra{n, \alpha, \beta}$.
        
        For every $2^{-\poly\rbra{n}} \leq \alpha < \beta \leq 1 - 2^{-\poly\rbra{n}}$ with $\beta - \alpha \geq 1/\poly\rbra{n}$, we reduce the problem $\textsc{PureFidTest}\rbra{n, \sqrt{1-\beta^2}, \sqrt{1-\alpha^2}}$ to $\textsc{TrDisTest}\rbra{n, \alpha, \beta}$ as follows.
        By Theorem \ref{thm:fidtest}, $\textsc{PureFidTest}\rbra{n, \sqrt{1-\beta^2}, \sqrt{1-\alpha^2}}$ is $\mathsf{BQP}$-complete since $2^{-\poly\rbra{n}} \leq \sqrt{1-\beta^2} < \sqrt{1-\alpha^2} \leq 1 - 2^{-\poly\rbra{n}}$ and $\sqrt{1-\alpha^2} - \sqrt{1-\beta^2} \geq 1/\poly\rbra{n}$.
        Suppose that quantum circuits $O_{\psi}$ and $O_{\phi}$ prepare pure quantum states $\ket{\psi}$ and $\ket{\phi}$, respectively, with a promise that either $\abs{\braket{\psi}{\phi}}^2 < \sqrt{1-\beta^2}$ or $\abs{\braket{\psi}{\phi}}^2 > \sqrt{1-\alpha^2}$.
        As $T\rbra{\ket{\psi}, \ket{\phi}} = \sqrt{1 - \abs{\braket{\psi}{\phi}}^2}$ by Eq. (\ref{eq:trace-distance-by-fidelity-pure}), we conclude that either $T\rbra{\ket{\psi}, \ket{\phi}} < \alpha$ or $T\rbra{\ket{\psi}, \ket{\phi}} > \beta$, promised that it is in either case.
        Therefore, we obtain a Karp reduction from $\textsc{PureFidTest}\rbra{n, \sqrt{1-\beta^2}, \sqrt{1-\alpha^2}}$ to $\textsc{TrDisTest}\rbra{n, \alpha, \beta}$.
        As a result, $\textsc{TrDisTest}\rbra{n, \alpha, \beta}$ is $\mathsf{BQP}$-hard. 
    \end{proof}
    
    \section{Approximately Low-Rank Quantum States} \label{sec:approx-low-rank}
    
    In this section, we discuss how our algorithms can be applied to approximately low-rank quantum states $\rho$ and $\sigma$. 
    Suppose we are given some prior knowledge $W_\rho\rbra*{\cdot}$ and $R_\rho\rbra*{\cdot}$ about the approximately low-rank quantum state $\rho$, defined as follows.

    \begin{definition} [$W_\rho$ and $R_\rho$]
    \label{def:W-R}
        For a quantum state $\rho$, we use two functions $W_\rho\rbra{\cdot}$ and $R_\rho\rbra{\cdot}$, with $W_\rho\rbra{\delta} \geq w\rbra{\rho, \delta}$ and  $R_\rho\rbra{\delta} \geq \rank_\delta\rbra{\rho}$ for all $\delta \geq 0$,
        to denote our prior knowledge of $\rho$.
    \end{definition}
    
    Let us start with identifying a class of approximately low-rank operators (see Section \ref{sec:approx-rank} for the notations used here). 
    
    \begin{definition} [Approximately low-rank operators]
        Let $r, \delta, \varepsilon \geq 0$.
        An Hermitian operator $A$ is said to be $\rbra*{r, \delta, \varepsilon}$-approximately-low-rank, if $\rank_{\delta}\rbra*{A} \leq r$ and $w\rbra*{A, \delta} \leq \varepsilon$.
    \end{definition}
    For every Hermitian operator $A$ of rank $r$, we note that $A$ is $\rbra*{r, 0, 0}$-approximately-low-rank, and also $\rbra*{r, \delta, r\delta}$-approximately-low-rank for every $\delta > 0$. This type of approximately low-rank quantum states were also considered in \cite{GP22} for low-rank fidelity estimation. 
    Intuitively, an $\rbra*{r, \delta, \varepsilon}$-approximately-low-rank quantum state $\rho$ is close to a quantum state of rank $r$ in the sense that: 
    \begin{enumerate}
        \item At most $r$ eigenvalues have absolute values greater than $\delta$; and 
        \item The sum of absolute eigenvalues that are not greater than $\delta$ is bounded by $\varepsilon$.
    \end{enumerate} Roughly speaking, there is a quantum state $\widetilde \rho$ of rank $r$ such that $\Abs*{\rho - \widetilde \rho} \leq \delta$ and $\tr\rbra*{\abs*{\rho - \widetilde \rho}} \leq \varepsilon$. 

    Note that in Theorem \ref{thm:purified} and Theorem \ref{thm:sample}, a condition $w\rbra*{\rbra*{\rho-\sigma}/2, \delta_p} \leq \varepsilon/4$ is required. 
    In the following, we will explain how to achieve this condition for approximately low-rank quantum states $\rho$ and $\sigma$. 
    Firstly, we show that the difference of two approximately low-rank quantum states is also approximately low-rank.
    
    \begin{proposition} \label{prop:app-low-rank}
        Suppose quantum state $\rho$ is $\rbra*{r_1,\delta,\varepsilon_1}$-approximately-low-rank and quantum state $\sigma$ is $\rbra*{r_2,\delta,\varepsilon_2}$-approximately-low-rank. Then, $\rbra*{\rho-\sigma}/2$ is $\rbra*{r_1+r_2,\delta/2,\rbra*{\rbra*{r_1+r_2}\delta+\varepsilon_1+\varepsilon_2}/2}$-approximately-low-rank.
    \end{proposition}
    
    \begin{proof}
        Let $\eta = \rho - \sigma$. Let the eigenvalues of $\rho$, $\sigma$ and $\eta$ be $\alpha_{i}$, $\beta_{i}$ and $\gamma_{i}$, respectively. We assume that $\alpha_i$, $\beta_i$ and $\gamma_i$ are non-increasing. Since $\rho$ is $\rbra*{r_1,\delta,\varepsilon_1}$-approximately-low-rank, then $\alpha_1 \geq \dots \geq \alpha_{r_1} > \delta \geq \alpha_{r_1+1} \geq \dots \geq \alpha_{N} \geq 0$ and $\sum_{j=r_{1}+1}^{N} \alpha_{j} \leq \varepsilon_1$. Since $\sigma$ is $\rbra*{r_2,\delta,\varepsilon_2}$-approximately-low-rank, then $\beta_1 \geq \dots \geq \beta_{r_2} > \delta \geq \beta_{r_2+1} \geq \dots \geq \beta_{N} \geq 0$ and $\sum_{j=r_{2}+1}^{N} \beta_{j} \leq \varepsilon_2$. 
        
        We only have to consider the case that $r_1+r_2 < N$. For every $r_1+1 \leq i \leq N-r_2$, by Weyl's theorem on eigenvalues \cite{Wey12}, we have
        \begin{equation}
            \alpha_{N} - \beta_{N-i+1} \leq \gamma_i \leq \alpha_{i} - \beta_{N},
        \end{equation}
        which gives $-\beta_{N-i+1} \leq \gamma_i \leq \alpha_i$ and thus $\abs*{\gamma_i} \leq \max\cbra*{\alpha_i, \beta_{N-i+1}} \leq \delta$. From this, it can be seen that $\rank_{\delta}\rbra*{\eta} \leq r_1+r_2$. Moreover,
        \begin{align}
            w\rbra*{\eta, \delta} 
            & = \sum_{j \colon \abs*{\gamma_j} \leq \delta} \abs*{\gamma_j} \\
            & = \sum_{j=1}^{r_1} \mathbbm{1}_{\abs*{\gamma_j} \leq \delta} \abs*{\gamma_j} + \sum_{j=N-r_2+1}^{N} \mathbbm{1}_{\abs*{\gamma_j} \leq \delta} \abs*{\gamma_j} + \sum_{j=r_1+1}^{N-r_2} \abs*{\gamma_j} \\
            & \leq r_1 \delta + r_2 \delta + \sum_{j=r_1+1}^{N-r_2} \max\cbra*{\alpha_j, \beta_{N-j+1}} \\
            & \leq r_1 \delta + r_2 \delta + \sum_{j=r_1+1}^{N-r_2} \alpha_j +  \sum_{j=r_1+1}^{N-r_2} \beta_{N-j+1} \\
            & \leq r_1 \delta + r_2 \delta + \varepsilon_1 + \varepsilon_2.
        \end{align}
        Therefore, $\eta$ is $\rbra*{r_1+r_2, \delta, r_1 \delta + r_2 \delta + \varepsilon_1 + \varepsilon_2}$-approximately-low-rank, and from this we see that $\eta/2 = \rbra*{\rho-\sigma}/2$ is $\rbra*{r_1+r_2, \delta/2, \rbra*{r_1 \delta + r_2 \delta + \varepsilon_1 + \varepsilon_2}/2}$-approximately-low-rank.
    \end{proof}

    Secondly, note that for every $\delta \geq 0$, $\rho$ is $\rbra*{R_{\rho}\rbra*{\delta}, \delta, W_{\rho}\rbra*{\delta}}$-approximately-low-rank. 
    For every desired precision $\varepsilon > 0$, choose $\delta_1$ and $\delta_2$ such that $W_{\rho}\rbra*{\delta_1} \leq \varepsilon/8$ and $W_{\sigma}\rbra*{\delta_2} \leq \varepsilon/8$. 
    Let $r_1 = R_{\rho}\rbra*{\delta_1}$ and $r_2 = R_{\sigma}\rbra*{\delta_2}$. 
    We take $\delta_p = 2\min\cbra*{\delta_1, \delta_2, \varepsilon/8r_1, \varepsilon/8r_2}$, then $\rho$ is $\rbra*{r_1, 2\delta_p, \varepsilon/8}$-approximately-low-rank and $\sigma$ is $\rbra*{r_2, 2\delta_p, \varepsilon/8}$-approximately-low-rank.
    By Proposition \ref{prop:app-low-rank}, it holds that $\rbra*{\rho-\sigma}/2$ is $\rbra*{r_1+r_2, \delta_p, \varepsilon/4}$-approximately-low-rank, which immediately yields the condition
    $w\rbra*{\rbra*{\rho-\sigma}/2, \delta_p} \leq \varepsilon/4$ required by Theorem \ref{thm:purified} and Theorem \ref{thm:sample}. 
    Therefore, we can apply Theorem \ref{thm:purified} to obtain a quantum algorithm with query complexity $\widetilde O\rbra*{\delta_p^{-1}\varepsilon^{-1}}$ given purified access, and apply Theorem \ref{thm:sample} to obtain a quantum algorithm with sample complexity $\widetilde O\rbra*{\delta_p^{-2}\varepsilon^{-3}}$ given identical copies. 

    In the following we give several examples of approximately low-rank trace distance estimation.
    The first one shows that the low-rank quantum states are just special cases of approximately low-rank quantum states,
    and previous results for low-rank states (Corollary \ref{corollary:purified} and Corollary \ref{corollary:sample}) can be recovered by applying theorems in this section.

    \begin{example} [Low-rank quantum states] \label{ex:low-rank-sate}
        Consider Problem~\ref{prb:main}, the trace distance estimation of two low-rank quantum states $\rho$ and $\sigma$ 
        with $\rank(\rho),\rank(\sigma)\leq r$. In this case, $\rho$ and $\sigma$ are also approximately low-rank in the sense that we have 
        $R_\rho\rbra{\delta} = R_\sigma\rbra{\delta} = r$ and $W_{\rho}\rbra{\delta} = W_{\sigma}\rbra{\delta} = r\delta$. For every desired precision $\varepsilon>0$, let $\delta_1 = \delta_2 = \varepsilon/8r$, and we have $\delta_p = \varepsilon/4r$. Therefore, we can apply Theorem \ref{thm:purified} to obtain a quantum algorithm with query complexity $\widetilde O\rbra{\delta_p^{-1}\varepsilon^{-1}} = \widetilde O\rbra{r \varepsilon^{-2}}$ given purified access; and apply Theorem \ref{thm:sample} to obtain a quantum algorithm with sample complexity $\widetilde O\rbra{\delta_p^{-2}\varepsilon^{-3}} = \widetilde O\rbra*{r^2 \varepsilon^{-5}}$ given identical copies. These results recover Corollary \ref{corollary:purified} and Corollary \ref{corollary:sample}.
    \end{example}

    The second example concerns the practical scenario when we prepare some low-rank quantum states but exposed to noise; in particular, a relatively small depolarizing noise is considered. Note that the noisy states are no longer low-rank but approximately low-rank.

    \begin{example} [Depolarizing channels] \label{ex:depolarizing-channels}
        Let $\mathcal{E}$ be a depolarizing channel acting on an $N$-dimensional Hilbert space, with parameter $\lambda > 0$:
        \begin{equation}
        \mathcal{E}\rbra*{\rho} = \rbra*{1 - \lambda} \rho + \lambda\frac{I}{N} .
        \end{equation}
        Our goal is to estimate the trace distance between $\mathcal{E}(\rho)$ and $\mathcal{E}(\sigma)$,
        where $\rank(\rho),\rank(\sigma)\leq r$. 
        Let $R\rbra*{\delta} := R_{\mathcal{E}\rbra*{\rho}}\rbra*{\delta} = R_{\mathcal{E}\rbra*{\sigma}}\rbra*{\delta}$ and $W\rbra*{\delta} := W_{\mathcal{E}\rbra*{\rho}}\rbra*{\delta} = W_{\mathcal{E}\rbra*{\sigma}}\rbra*{\delta}$, then
        \begin{equation}
        R\rbra*{\delta} = \begin{cases}
            r, & \delta \geq \frac{\lambda}{N}, \\
            N, & \textup{otherwise},
        \end{cases}
        \end{equation}
        \begin{equation}
        W\rbra*{\delta} = \begin{cases}
            \frac{\lambda}{N}\rbra*{N-r} + r\delta, & \delta \geq \frac{\lambda}{N}, \\
            N\delta, & \textup{otherwise}.
        \end{cases}
        \end{equation}
        When $\lambda$ is relatively small;
        that is, when the precision $\varepsilon \gg \lambda$,
        one can choose $\delta_1 = \delta_2 = \frac{\varepsilon - 8\lambda}{8r} + \frac{\lambda}{N} = \Theta\rbra*{ \frac{\varepsilon}{r} + \frac{\lambda}{N} }$ satisfying $W\rbra*{\delta_1} = W\rbra*{\delta_2} \leq \varepsilon/8$. Note that $r_1 = R\rbra*{\delta_1} = r_2 = R\rbra*{\delta_2} = r$ (because $\delta_1 = \delta_2 \geq \frac{\lambda}{N}$), and thus $\delta_p = 2\min\cbra*{\delta_1, \delta_2, \varepsilon/8r_1, \varepsilon/8r_2} = \Theta\rbra{\varepsilon/r}$.
        In this case, our quantum algorithms can estimate the trace distance between ${\mathcal{E}\rbra*{\rho}}$ and ${\mathcal{E}\rbra*{\sigma}}$ within additive error $\varepsilon$ with the same complexity as that in Example \ref{ex:low-rank-sate} for low-rank quantum states. 
    \end{example}

    The next example considers estimating the trace distance between the Gibbs states of gapped Hamiltonians. 

    \begin{example} [Gibbs states of gapped Hamiltonians]
        Suppose that $H$ (resp. $G$) is an $N$-dimensional Hamiltonian with a gap $\Delta$ between the $k$-th and the $\rbra{k+1}$-th smallest eigenvalues of $H$ (resp. $G$).
        Let $\rho = \exp\rbra{-H}/\tr\rbra{\exp\rbra{-H}}$ and $\sigma = \exp\rbra{-G}/\tr\rbra{\exp\rbra{-G}}$ be the Gibbs states of $H$ and $G$, respectively. 
        Let $R\rbra*{\delta} := R_{\rho}\rbra*{\delta} = R_{\sigma}\rbra*{\delta}$ and $W\rbra*{\delta} := W_{\rho}\rbra*{\delta} = W_{\sigma}\rbra*{\delta}$, then:
        \begin{equation}
        R\rbra*{\delta} = \begin{cases}
            k, & \delta > \rbra*{\exp\rbra{\Delta} k+1}^{-1}, \\
            N, & \textup{otherwise},
        \end{cases}
        \end{equation}
        \begin{equation}
        W\rbra*{\delta} = \begin{cases}
            \frac{N-k}{\exp\rbra{\Delta}k +1} + k\delta, & \delta > \rbra*{\exp\rbra{\Delta} k+1}^{-1}, \\
            N\delta, & \textup{otherwise}.
        \end{cases}
        \end{equation}
        Suppose the desired precision is $\varepsilon \gg \exp\rbra*{-\Delta} N/k$; this lower bound can be small for large gap $\Delta$.
        We can choose 
        \begin{align}
        \delta_1 = \delta_2 & = \frac{1}{k} \rbra*{ \frac{\varepsilon}{8} - \frac{N-k}{\exp\rbra{\Delta}k + 1} } \\
        & = \Theta\rbra*{\frac{\varepsilon}{k}} \gg \frac{1}{\exp\rbra*{\Delta}k + 1},
        \end{align}
        which satisfies $W\rbra*{\delta_1} = W\rbra*{\delta_2} \leq \varepsilon/8$, and $r_1 = R\rbra*{\delta_1} = r_2 = R\rbra*{\delta_2} = k$; 
        then we have $\delta_p = 2\min\cbra*{\delta_1, \delta_2, \varepsilon/8r_1, \varepsilon/8r_2} = \Theta\rbra{\varepsilon/k}$.
        With these, we can apply Theorem \ref{thm:purified} to obtain a quantum algorithm with query complexity $\widetilde O\rbra{\delta_p^{-1}\varepsilon^{-1}} = \widetilde O\rbra{k\varepsilon^{-2}}$ given purified access, and apply Theorem \ref{thm:sample} to obtain a quantum algorithm with sample complexity $\widetilde O\rbra{\delta_p^{-2}\varepsilon^{-3}} = \widetilde O\rbra{k^2\varepsilon^{-5}}$ given identical copies. 
        The result is similar to that for low-rank quantum states (Example \ref{ex:low-rank-sate} and Example \ref{ex:depolarizing-channels}).
    \end{example}

    At last, we give an artificial example that can be solved by our quantum algorithms, where quantum states are no longer related to low-rank conditions but the eigenvalues of quantum states have certain upper bounds. 
    This non-trivial example shows that our algorithms have the potential to be applied to more cases where quantum states are not low-rank. 

    \begin{example}
        Suppose two $N$-dimensional quantum states $\rho$ and $\sigma$ have eigenvalues $\alpha_1 \geq \dots \geq \alpha_{N}$ and $\beta_1 \geq \dots \geq \beta_{N}$, respectively, where $\max\cbra{\alpha_i, \beta_i} \leq C/i^2$ for some constant $C > 0$. 
        Let $R\rbra*{\delta} := R_{\rho}\rbra*{\delta} = R_{\sigma}\rbra*{\delta}$ and $W\rbra*{\delta} := W_{\rho}\rbra*{\delta} = W_{\sigma}\rbra*{\delta}$, then:
        \begin{equation}
        R\rbra*{\delta} = \sqrt{\frac{C}{\delta}}, 
        \end{equation}
        \begin{equation}
        W\rbra*{\delta} = \sum_{i=\ceil*{\sqrt{C/\delta}}}^{N} \frac{C}{i^2} \leq \frac{C}{\sqrt{C/\delta}-1} - \frac{C}{N-1}.
        \end{equation}
        For every desired precision $\varepsilon > 0$, we can choose 
        \begin{equation}
            \delta_1 = \delta_2 = \frac{C \rbra*{\frac{\varepsilon}{8} + \frac{C}{N-1}}^2}{\rbra*{\frac{\varepsilon}{8} + \frac{C}{N-1} + C}^2} = \Theta\rbra*{\varepsilon^2},
        \end{equation}
        which satisfies $W\rbra{\delta_1} = W\rbra*{\delta_2} \leq \varepsilon/8$.
        Note that $r_1 = r_2 = R\rbra{\delta_1} = R\rbra{\delta_2} = \Theta\rbra{\varepsilon^{-1}}$, and therefore $\delta_p = 2\min\cbra*{\delta_1, \delta_2, \varepsilon/8r_1, \varepsilon/8r_2} = \Theta\rbra{\varepsilon^2}$.
        With these, we can apply Theorem \ref{thm:purified} to obtain a quantum algorithm with query complexity $\widetilde O\rbra{\delta_p^{-1}\varepsilon^{-1}} = \widetilde O\rbra{\varepsilon^{-3}}$ given purified access, and apply Theorem \ref{thm:sample} to obtain a quantum algorithm with sample complexity $\widetilde O\rbra{\delta_p^{-2}\varepsilon^{-3}} = \widetilde O\rbra{\varepsilon^{-7}}$ given identical copies. 
    \end{example}

    \section*{Acknowledgment}
    
    We would like to thank John Wright for explaining the quantum algorithms for quantum state certification in \cite{BOW19}. 
    We would like to thank Andr{\'{a}}s Gily{\'{e}}n and Alexander Poremba for communication regarding the related work \cite{GP22}.
    We would like to thank Yupan Liu for pointing out the $\mathsf{BQP}$-completeness of pure-mixed fidelity estimation shown in \cite{RASW23}, and pointing out the related work \cite{CCC19} which conjectured that low-rank trace distance estimation is in $\mathsf{BQP}$. 
    We would like to thank Ansis Rosmanis for pointing out that, given purified access, the SWAP test \cite{BCWdW01} can estimate the fidelity between pure quantum states within additive error $\varepsilon$ using $O\rbra*{1/\varepsilon}$ queries equipped with quantum amplitude estimation \cite{BHMT02}.
    We would like to thank Mark M.\ Wilde for helpful comments.

    Qisheng Wang was supported by the MEXT Quantum Leap Flagship Program (MEXT Q-LEAP) grants \mbox{No.~JPMXS0120319794}. Zhicheng Zhang was supported by the Sydney Quantum Academy, Sydney, NSW, Australia.

\addcontentsline{toc}{section}{References}

\bibliographystyle{unsrturl}
\bibliography{main}
    
    \appendix
    
    \section{Trace Distance Estimation via the SWAP Test} \label{app:estimation-via-swap}
    
    Our algorithm can estimate the trace distance between pure quantum states (i.e., $r = 1$) with query complexity $\widetilde O\rbra*{1/\varepsilon^2}$, given purified access, which matches the query complexity $O\rbra*{1/\varepsilon^2}$ by the SWAP test \cite{BCWdW01}. 
    To see this, suppose two pure quantum states $\ket{\psi}$ and $\ket{\phi}$ are given by two quantum unitary operators $U_{\psi}$ and $U_{\phi}$ such that $U_{\psi}\ket{0} = \ket{\psi}$ and $U_{\phi}\ket{0} = \ket{\phi}$. 
    By the SWAP test \cite{BCWdW01} and quantum amplitude estimation (Theorem \ref{thm:amp-estimation}), we can estimate  $\abs*{\braket{\psi}{\phi}}^2$ within additive error $\delta$ using $O\rbra*{1/\delta}$ queries to $U_\psi$ and $U_\phi$.
    That is, we can obtain $\widetilde x$ with high probability such that
    \begin{equation}
    \abs*{\widetilde x - \abs*{\braket{\psi}{\phi}}^2} \leq \delta.
    \end{equation}
    Note that the trace distance between pure quantum states $\ket{\psi}$ and $\ket{\phi}$ is given by 
    \begin{equation}
    T\rbra*{\ket{\psi}, \ket{\phi}} = \sqrt{1 - \abs*{\braket{\psi}{\phi}}^2}.
    \end{equation}
    Following this formula, we can estimate the trace distance $T\rbra*{\ket{\psi}, \ket{\phi}}$ by $\sqrt{1 - \widetilde x}$. 
    
    \begin{proposition}
    With high probability, the error is bounded by
    \begin{equation}
        \abs*{\sqrt{1 - \widetilde x} - T\rbra*{\ket{\psi}, \ket{\phi}}} \leq 2 \sqrt{\delta}.
    \end{equation}
    \end{proposition}
    \begin{proof}
        We consider two cases.
        \begin{enumerate}
            \item $\min\cbra*{\widetilde x, \abs*{\braket{\psi}{\phi}}^2} \leq 1-\delta$. In this case, $\max \cbra*{ \sqrt{1-\widetilde x}, \sqrt{1-\abs*{\braket{\psi}{\phi}}^2} } \geq \sqrt{\delta}$. We have
            \begin{align}
                \abs*{\sqrt{1 - \widetilde x} - T\rbra*{\ket{\psi}, \ket{\phi}}}
                & = \abs*{ \frac{\widetilde x - \abs*{\braket{\psi}{\phi}}^2}{\sqrt{1-\widetilde x} + \sqrt{1-\abs*{\braket{\psi}{\phi}}^2}} } \\
                & \leq \frac{\abs*{ \widetilde x - \abs*{\braket{\psi}{\phi}}^2 }}{\sqrt{\delta}} \\
                & \leq \frac{\delta}{\sqrt{\delta}} = \sqrt{\delta}.
            \end{align}
            
            \item $\min\cbra*{\widetilde x, \abs*{\braket{\psi}{\phi}}^2} > 1-\delta$. In this case, $\max \cbra*{ \sqrt{1-\widetilde x}, \sqrt{1-\abs*{\braket{\psi}{\phi}}^2} } < \sqrt{\delta}$.
            \begin{align}
                \abs*{\sqrt{1 - \widetilde x} - T\rbra*{\ket{\psi}, \ket{\phi}}}
                & = \abs*{\sqrt{1-\widetilde x}} + \abs*{\sqrt{1-\abs*{\braket{\psi}{\phi}}^2}} \\
                & \leq 2\sqrt{\delta}.
            \end{align}
        \end{enumerate}
        The both cases together yield the proof.
    \end{proof}
    Finally, by taking $\delta = \varepsilon^2/4$, we can estimate the trace distance $T\rbra*{\ket{\psi}, \ket{\phi}}$ within additive error $\varepsilon$ using $O\rbra*{1/\delta} = O\rbra*{1/\varepsilon^2}$ queries to $U_\psi$ and $U_\phi$. 
    In the same way, if only identical copies of $\ket{\psi}$ and $\ket{\phi}$ are given, we can estimate their trace distance within additive error $\varepsilon$ using $O\rbra*{1/\varepsilon^4}$ samples of them.
    We explicitly state these simple results as follows.
    
    \begin{theorem} [Trace distance estimation for pure quantum states]
        There is a quantum algorithm that estimates the trace distance of two pure quantum states within additive error $\varepsilon$, 
        \begin{itemize}
            \item Given purified access, with query complexity $O\rbra*{1/\varepsilon^2}$, and time complexity $O\rbra*{1/\varepsilon^2 \cdot \log\rbra*{N} }$; and
            \item Given sample access, with sample complexity $O\rbra*{1/\varepsilon^4}$, time complexity $O\rbra*{1/\varepsilon^4 \cdot \log\rbra*{N} }$, and depth complexity $O\rbra*{1}$. 
        \end{itemize}
    \end{theorem}

\end{document}